\documentclass[svgnames]{llncs}
\pdfoutput = 1

\usepackage[hidelinks]{hyperref}
\usepackage{amsmath,amssymb}
\usepackage{amsfonts}
\usepackage[english]{babel}
\usepackage{enumerate}

\usepackage{tikz}
\usetikzlibrary{arrows} 
\usepackage{oplotsymbl}

\usepackage{ifthen} 
\newcommand{\optionalargumentinbrackets}[1]{\ifthenelse{\equal {#1} {} }{}{~(#1)}}

\newcounter{theoremcnt}[section]
\renewcommand{\thetheoremcnt}{\thesection.\arabic{theoremcnt}}

\renewenvironment{definition}[1][]
{\begin{trivlist}\refstepcounter{theoremcnt}
\item[]\textbf{Definition~\thetheoremcnt \optionalargumentinbrackets{#1}.} }{\end{trivlist}}

{\begin{trivlist}\refstepcounter{theoremcnt}
\item[]\textbf{Example~\thetheoremcnt \optionalargumentinbrackets{#1}.} }{\end{trivlist}}

\renewenvironment{lemma}[1][]%
{\begin{trivlist}\refstepcounter{theoremcnt}
\item[]\textbf{Lemma~\thetheoremcnt \optionalargumentinbrackets{#1}.} }{\end{trivlist}}

\renewenvironment{proof}[1][]%
{\begin{trivlist}\item[]\textbf{Proof\optionalargumentinbrackets{#1}.} }{\qed\end{trivlist}}

\renewenvironment{theorem}[1][]%
{\begin{trivlist}\refstepcounter{theoremcnt}
\item[]\textbf{Theorem~\thetheoremcnt \optionalargumentinbrackets{#1}.} }{\end{trivlist}}

{\begin{trivlist}\refstepcounter{theoremcnt}
\item[]\textbf{Corollary~\thetheoremcnt \optionalargumentinbrackets{#1}.} }{\end{trivlist}}

\makeatletter
\newcommand{\xRightarrow}[2][]{\ext@arrow
  0359\Rightarrowfill@{#1}{#2}}
\makeatother

\newcommand{\Arrow}[1]{\xRightarrow{\, {#1} \; \,}}
\newcommand{\AXd}{\mathit{A \mkern-2mu X}}
\newcommand{\AXdb}{\mathit{A \mkern-2mu X}^b}
\newcommand{\AXp}{\mathit{A \mkern-2mu X}_{\mkern-3mu p}}
\newcommand{\AXpb}{\mathit{A \mkern-2mu X}^b_{\mkern-3mu p}}
\newcommand{\DistrE}{\Distr(\calE)}
\newcommand{\DistrEconc}{\Distr(\calE_{\mathit{cc}})}
\newcommand{\DistrX}{\Distr(X)}

\newcommand{\alphaarrow}{\arrow{\alpha}}
\newcommand{\alphahidearrow}{\arrow{(\alpha)}}
\newcommand{\arrow}[1]{\xrightarrow{\raisebox{0ex}[0.1ex][-0.1ex]{\scriptsize $\,#1\,$}}}
\newcommand{\astop}{a \pref \partial(\bfzero)}
\newcommand{\bfzero}{\normalfont{\textbf{0}}}
\newcommand{\bigoplusiinI}{\textstyle{\bigoplus_{i \mathord{\in} I}}}
\newcommand{\bigoplusjinJ}{\textstyle{\bigoplus_{j \mathord{\in} J}}}

\newcommand{\bigopluskinK}{\textstyle{\bigoplus_{k \mathord{\in} K}}}
\newcommand{\bisim}{\mathrel{\,%
  \raisebox{.3ex}{$\underline{\makebox[.7em]{$\leftrightarrow$}}$}\,}}
\newcommand{\bnfeq}{\mathrel{{:}{:}{=}}}
\newcommand{\brbisim}{\mathrel{\,%
  \raisebox{.3ex}{$\underline{\makebox[.8em]{$\leftrightarrow$}}_{\,
      b}$}\mkern-1mu}}
\newcommand{\bstop}{b \pref \partial(\bfzero)}
\newcommand{\der}{\textsl{der} \mkern1mu}  

\newcommand{\half}{\textstyle{\frac12}}

\newcommand{\lc}{\lbrace \:}
\newcommand{\mopcalR}{\,{\calR}\,}   
\newcommand{\mycdot}{\mathop{\cdot}}
\newcommand{\myfrac}[2]{\textstyle{\frac{#1}{#2}}}

\newcommand{\nbrbisim}{\mkern4mu {\not\mkern-2.3mu\brbisim} \mkern4mu}
\newcommand{\nrootedbrbisim}{\mkern4mu {\not\mkern-1.5mu\rootedbrbisim} \mkern4mu}
\newcommand{\oplusr}{\oplus_r}
\newcommand{\pref}{\mathop{.}}

\newcommand{\rc}{\: \rbrace}
\newcommand{\rootedbrbisim}{\mathrel{\,%
  \raisebox{.3ex}{$\underline{\makebox[.7em]{$\leftrightarrow$}}_{\,
      rb}$}\mkern-1mu}}
\newcommand{\singleton}[1]{\lbrace {#1} \rbrace}
\newcommand{\sosrule}[2]{%
  \def\arraystretch{0.50}
  \begin{array}{c} {#1}\rule[-4pt]{0pt}{11pt} \\ \hline
    \rule{0pt}{13pt}{#2} \end{array}
  \def\arraystretch{1.0}}
\newcommand{\sumiinI}{\textstyle\sum_{i {\in} I}}
\newcommand{\sumjinJ}{\textstyle\sum_{j {\in} J}}
\newcommand{\sumkinK}{\textstyle\sum_{k {\in} K}}

\newcommand{\tauarrow}{\arrow{\tau}}
\newcommand{\tauhidearrow}{\arrow{(\tau)}}

\newcommand{\Cbar}{\bar{C}}
\newcommand{\Ebar}{\bar{E}}

\newcommand{\Fbar}{\bar{F}}
\newcommand{\Gbar}{\bar{G}}
\newcommand{\Pbar}{\bar{P}}
\newcommand{\Qbar}{\bar{Q}}
\newcommand{\etabar}{\bar{\eta}}

\newcommand{\nubar}{\bar{\nu}}

\newcommand{\Distr}{\textsl{Distr}\mkern1mu}

\newcommand{\calA}{\mathcal{A}}

\newcommand{\calE}{\mathcal{E}}
\newcommand{\calEconc}{\mathcal{E}_{\mathit{cc}}}

\newcommand{\calP}{\mathcal{P}}
\newcommand{\calPconc}{\mathcal{P}_{\!\mathit{cc}}}
\newcommand{\calR}{\mathcal{R}}

\newcommand{\blankline}{\vspace*{1.0\baselineskip}}
\newcommand{\halflineup}{\vspace*{-0.5\baselineskip}}

\colorlet{colorA}{LightYellow}
\colorlet{colorB}{Pink}
\colorlet{colorC}{LightBlue}
\colorlet{colorD}{Plum}
\colorlet{colorE}{Gold!50}
\colorlet{colorF}{Orange!50}
\colorlet{colorG}{Red!50}
\colorlet{colorH}{Green!50}
\colorlet{colorI}{NavyBlue!50}
\colorlet{colorJ}{Brown!50}
\colorlet{colorK}{Coral}
\colorlet{colorL}{RoyalBlue!50}
\colorlet{colorM}{OliveDrab!50}
\colorlet{colorX}{Gray!50}

\newcommand{\bisimilarity}{bisimilarity}
\newcommand{\probc}[1]{\mathbin{\mbox{${}_{\scriptscriptstyle #1 \mkern1mu} \hspace{-.8pt}\oplus\,$}}}
\renewcommand{\oplusr}{\probc{r}}
\newcommand{\den}[1]{\mbox{$[\hspace{-1.6pt}[$}#1\mbox{$]\hspace{-1.6pt}]$}}  
\newcommand{\plat}[1]{\raisebox{0pt}[0pt][0pt]{#1}}
\definecolor[named]{darkgreen}{cmyk}{1,0,.7,.5}
\definecolor[named]{darkorange}{cmyk}{0,0.82,1,0.01}
\definecolor[named]{purpleblue}{cmyk}{.42,.85,0,.29}
\newcommand{\dg}[1]{\textcolor{darkgreen}#1}
\newcommand{\dor}[1]{\textcolor{darkorange}#1}
\newcommand{\mg}[1]{\textcolor{purpleblue}#1}
\newcommand{\weg}[1]{}

\title{A Complete Axiomatization of\texorpdfstring{\\}{} Branching
  Bisimilarity for
  a Simple Process Language with Probabilistic Choice}
\author{R.J.~van Glabbeek${}^{\,1,2}$ \and
  J.F.~Groote${}^{\,3}$ \and E.P.~de Vink${}^{\,3}$}
\institute{Data61, CSIRO, Sydney, Australia \\
  \and
  Computer Science and Engineering, University of New South Wales,
  Australia 
  \and
  Department of Mathematics and Computer Science,\\
  Eindhoven University of Technology, Eindhoven, The Netherlands \\
  \email{rvg@cs.stanford.edu}, \email{J.F.Groote@tue.nl},
  \email{evink@win.tue.nl}} 
\date{}

\begin{document}

\titlerunning{A complete axiomatization for branching {\bisimilarity}}
\authorrunning{Van~Glabbeek, Groote \& De~Vink}
\maketitle

\begin{abstract}
  \noindent
  This paper proposes a notion of branching {\bisimilarity} for
  non-deterministic probabilistic processes. In order to characterize
  the corresponding notion of rooted branching probabilistic
  {\bisimilarity}, an equational theory is proposed for a basic,
  recursion-free process language with non-deterministic as well as
  probabilistic choice. The proof of completeness of the
  axiomatization builds on the completeness of strong probabilistic
  {\bisimilarity} on the one hand and on the notion of a concrete
  process, i.e.\ a process that does not display  (partially)
  inert $\tau$-moves, on the other hand. The approach is first presented for the
  non-deterministic fragment of the calculus and next generalized to
  incorporate probabilistic choice, too.
  \vspace{2ex}

  \textit{This paper is dedicated to Catuscia Palamidessi, on the occasion of her 60th birthday. An
    extended abstract appears in \cite{GGV19}.}
\end{abstract}




\section{Introduction}

In~\cite{DP07:tcs}, in a setting of a process language featuring both
non-deterministic and probabilistic choice, Yuxin Deng and Catuscia
Palamidessi propose an equational theory for a notion of weak
{\bisimilarity} and prove its soundness and completeness. Not
surprisingly, the axioms dealing with a silent step are reminiscent to
the well-known $\tau$-laws of Milner~\cite{Mil89:phi,Mil89:ic}. The process
language treated in~\cite{DP07:tcs} includes recursion, thereby
extending the calculus and axiomatization of~\cite{BS01:icalp}. While
the weak transitions of~\cite{DP07:tcs} can be characterized as
finitary, infinitary semantics is treated in~\cite{FG19:jlamp},
providing a sound and complete axiomatization also building
on the seminal work of Milner~\cite{Mil89:ic}.

In this paper we focus on branching {\bisimilarity} in the sense
of~\cite{GW96:jacm}, rather than on weak {\bisimilarity} as
in~\cite{BS01:icalp,DP07:tcs,FG19:jlamp}. 
In the non-probabilistic setting branching {\bisimilarity} has the
advantage over weak {\bisimilarity} that it has far more efficient
algorithms~\cite{GJKW17:tcl,GV90:icalp}. Furthermore, it has a strong
logical underpinning~\cite{DV95:jacm}. It would be very attractive to
have these advantages available also in the probabilistic case, where
model checking is more demanding. See also the initial work reported
in~\cite{GV17:festschrift}.

For a similarly basic process language as in~\cite{DP07:tcs}, without
recursion though, we propose a notion of branching probabilistic {\bisimilarity}
as well as a sound and complete equational axiomatization. Hence,
instead of lifting 
all $\tau$-laws to the probabilistic setting, we only need to do this
for the B-axiom of~\cite{GW96:jacm}, the single axiom capturing inert
silent steps. For what is referred to as the alternating
model~\cite{HJ90:rtss}, branching probabilistic {\bisimilarity} has been
studied in~\cite{AGT12:tcs,AW06:tcs}. Also~\cite{Seg95:thesis}
discusses branching probabilistic {\bisimilarity}. However, the proposed
notions of branching {\bisimilarity} are either no congruence for the
parallel operator, or they invalidate the identities below which we
desire. The paper~\cite{AG09:sofsem} proposes a complete theory for a
variant of branching {\bisimilarity} that is not consistent with the
first $\tau$-law unfortunately.

Our investigation is led by the wish to identify the three processes
below, involving as a subprocess a probabilistic choice between $P$
and~$Q$. Essentially, ignoring the occurrence of the action~$a$
involved, the three processes represent (i)~a probabilistic choice of
weight~$\frac34$ between to instances of the subprocess mentioned,
(ii)~the subprocess on its own, and (iii)~a probabilistic choice of
weight~$\frac13$ for the subprocess and a rescaling of the subprocess,
in part overlapping.

\medskip

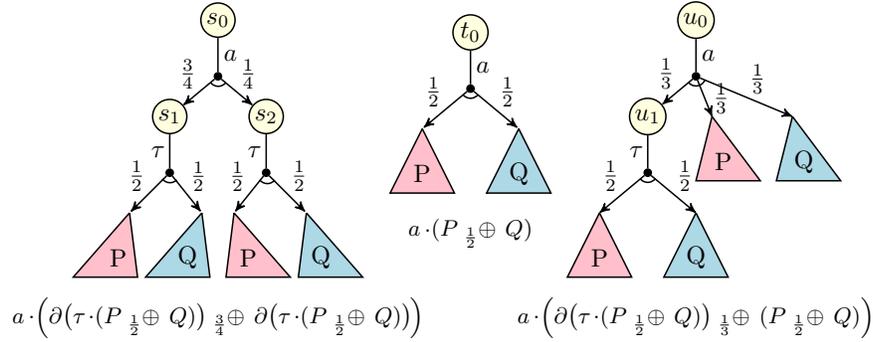
\scalebox{0.95}{%
\!\!\!


\hbox{%
\begin{tikzpicture}[> = stealth', semithick, inner sep = 1pt, scale=0.45]
  \tikzstyle{state} = [circle, draw=Black, fill=LightYellow, minimum size=4.0mm]
  \tikzstyle{split} = [circle, draw=black, fill=black]
  \node [state] (t0) at (2.5,9) {$s_0$}; 
  \node [split] (y0) at (2.5,7.25) {};
  \node [state] (t1) at (1.0,6) {$s_1$}; 
  \node [split] (y1) at (1.0,4.25) {}; 
  \coordinate (t3) at (-0.25,3); 
  \coordinate (t4) at (2.0,3); 
  \node [state] (t2) at (4.0,6) {$s_2$}; 
  \node [split] (y2) at (4.0,4.25) {}; 
  \coordinate (t5) at (3.0,3); 
  \coordinate (t6) at (5.25,3); 
  \path [>=] (t0) edge node [midway, right=1pt] {$a$} (y0);
  \path [->] (y0) edge node [midway, above left] {$\frac{3}{4}$} (t1);
  \path [->] (y0) edge node [midway, above right] {$\frac{1}{4}$} (t2);
  \draw [>=] (y0) ++(220:.3) arc (217.5:322.5:.3);
  \path [>=] (t1) edge node [midway, left=0.5pt] {$\tau$} (y1);
  \path [->, shorten >=0.25mm] (y1) edge node [midway, above left] {$\frac{1}{2}$} (t3);
  \path [->, shorten >=0.25mm] (y1) edge node [midway, above right] {$\frac{1}{2}$} (t4);
  \draw [>=] (y1) ++(220:.3) arc (217.5:315:.3);
  \fill [fill=Pink, draw=Black, >=] (t3) -- ++(-1.75,-2.0) -- ++(2.0,0.0)
  -- (t3);  
  \node at ([shift={(-0.35,-1.4)}] t3) {P};
  \fill [fill=LightBlue, draw=Black, >=] (t4) -- ++(-1.75,-2.0) --
  ++(2.0,0.0) -- (t4); 
  \node at ([shift={(-0.45,-1.4)}] t4) {Q};
  \path [>=] (t2) edge node [midway, left=0.5pt] {$\tau$} (y2);
  \path [->, shorten >=0.25mm] (y2) edge node [midway, above left] {$\frac{1}{2}$} (t5);
  \path [->, shorten >=0.25mm] (y2) edge node [midway, above right] {$\frac{1}{2}$} (t6);
  \draw [>=] (y2) ++(227.5:.3) arc (217.5:320:.3);
  \fill [fill=Pink, draw=Black, >=] (t5) -- ++(-0.25,-2.0) -- ++(2.0,0.0)
  -- (t5);  
  \node at ([shift={(+0.45,-1.4)}] t5) {P};
  \fill [fill=LightBlue, draw=Black, >=] (t6) -- ++(-0.25,-2.0) --
  ++(2.0,0.0) -- (t6); 
  \node at ([shift={(+0.45,-1.4)}] t6) {Q};
  \node (S) at (2.5,-0.25) {\small
    $a \pref 
    \Bigl( 
    \partial \bigl( \tau \pref (P \probc{\frac12} Q) \bigr) 
    \probc{\frac34} 
    \partial \bigl( \tau \pref (P \probc{\frac12} Q) \bigr)
    \mkern-2mu \Bigr)$} ;
\end{tikzpicture}
\hspace*{-0.75cm}
\begin{tikzpicture}[> = stealth' , semithick, inner sep = 1pt, scale=0.45]
  \tikzstyle{state} = [circle, draw=Black, fill=LightYellow, minimum size=4.0mm]
  \tikzstyle{split} = [circle, draw=black, fill=black]
  \node [state] (u0) at (2.5,9) {$t_0$}; 
  \node [split] (x0) at (2.5,7.25) {};
  \coordinate (u1) at (1.0,6); 
  \coordinate (u2) at (4.0,6); 
  \coordinate (u3) at (4.0,1); 
  \path (u2) -- (u3);
  \path [>=] (u0) edge node [midway, right=1pt] {$a$} (x0);
  \path [->, shorten >=0.25mm] (x0) edge node [midway, above left] {$\frac{1}{2}$} (u1);
  \path [->, shorten >=0.25mm] (x0) edge node [midway, above right] {$\frac{1}{2}$} (u2);
  \draw [>=] (x0) ++(217.5:.3) arc (217.5:322.5:.3);
  \fill [fill=Pink, draw=Black, >=] (u1) -- ++(-1.0,-2.0) --
  ++(2.0,0.0) -- (u1);  
  \node at ([shift={(0,-1.4)}] u1) {P};
  \fill[fill=LightBlue, draw=Black, >=] (u2) -- ++(-1.0,-2.0) -- ++(2.0,0.0) -- (u2);
  \node at ([shift={(0,-1.4)}] u2) {Q};
  \node (T) at (2.5,2.75) {\small
    $a \pref (P \probc{\frac12} Q)$} ; 
  \node (T') at (2.5,-0.5) {} ; 
\end{tikzpicture}
  \hspace*{-0.75cm}
\begin{tikzpicture}[> = stealth', semithick, inner sep = 1pt, scale=0.45]
  \tikzstyle{state} = [circle, draw=Black, fill=LightYellow, minimum size=4.0mm]
  \tikzstyle{split} = [circle, draw=Black, fill=black]
  \node [state] (s0) at (2.5,9) {$u_0$}; 
  \node [split] (x0) at (2.5,7.25) {};
  \node [state] (s1) at (1.0,6) {$u_1$}; 
  \node [split] (x1) at (1.0,4.25) {}; 
  \coordinate (s2) at (3,6); 
  \coordinate (s3) at (5.5,6); 
  \coordinate (s4) at (-0.5,3); 
  \coordinate (s5) at (2.5,3); 
  \path [>=] (s0) edge node [midway, right=1pt] {$a$} (x0);
  \path [->] (x0) edge node [midway, above left] {$\frac{1}{3}$} (s1);
  \path [->] (x0) edge node [midway, right=2pt] {$\frac{1}{3}$} (s2);
  \path [->] (x0) edge node [midway, above right] {$\frac{1}{3}$} (s3);
  \draw [>=] (x0) ++(217.5:.3) arc (217.5:340:.3);
  \path [>=] (s1) edge node [midway, left =0.5pt] {$\tau$} (x1);
  \fill [fill=Pink, draw=Black, >=] (s2) -- ++(-0.5,-2.0) -- ++(2.0,0.0) -- (s2); 
  \node at ([shift={(+0.35,-1.4)}] s2) {P};
  \fill [fill=LightBlue, draw=Black, >=] (s3) -- ++(-0.5,-2.0) --
  ++(2.0,0.0) -- (s3); 
  \node at ([shift={(+0.35,-1.4)}] s3) {Q};
  \draw [>=] (x1) ++(217.5:.3) arc (217.5:322.5:.3);
  \path [->] (x1) edge node [midway, above left] {$\frac{1}{2}$} (s4);
  \path [->] (x1) edge node [midway, above right] {$\frac{1}{2}$} (s5);
  \fill [fill=Pink, draw=Black, >=] (s4) -- ++(-1.0,-2.0) -- ++(2.0,0.0) -- (s4); 
  \node at ([shift={(0,-1.4)}] s4) {P};
  \fill [fill=LightBlue, draw=Black, >=] (s5) -- ++(-1.0,-2.0) --
  ++(2.0,0.0) -- (s5); 
  \node at ([shift={(0,-1.4)}] s5) {Q};
  \node (U) at (2.5,-0.25) {\small
    $a \pref 
    \Bigl( 
    \partial \bigl( \tau \pref (P \probc{\frac12} Q) \bigr)
    \probc{\frac13} 
    (P \probc{\frac12} Q)
    \Bigr)$} ;
\end{tikzpicture}
} 

\!\!%
}%

\medskip

\noindent
In our view, all three processes starting from $s_0$, $t_0$, and $u_0$
are equivalent. The behavior that can be observed from them when
ignoring $\tau$-steps and coin tosses to resolve probabilistic choices
is the same. This leads to a definition of probabilistic branching
{\bisimilarity} that hitherto was not proposed in the literature and
appears to be the pendant of weak distribution {\bisimilarity} defined
by~\cite{EHKTZ13:qest}.

As for~\cite{DP07:tcs} we seek to stay close to the treatment of the
non-deterministic fragment of the process calculus at hand. However,
as an alternate route in proving completeness, we rely on the
definition of a concrete process.  We first apply the approach for
strictly non-deterministic processes and \emph{mutatis mutandis} for
the whole language allowing processes that involve both
non-deterministic and probabilistic choice. For now, let's call a
process concrete if it doesn't exhibit inert transitions,
i.e.\ $\tau$-transitions that don't change the potential behavior of
the process essentially. The approach we follow first establishes
soundness for branching (probabilistic) {\bisimilarity} and soundness and
completeness for strong (probabilistic) {\bisimilarity}. Because of the
non-inertness of the silent steps involved, strong and branching
{\bisimilarity} coincide for concrete processes. The trick then is to
relate a pair of branching (probabilistically) bisimilar processes to
a corresponding pair of concrete processes. Since these are also
branching (probabilistically) bisimilar as argued, they are
consequently strongly (probabilistically) bisimilar, and, voil\`a,
provably equal by the completeness result for strong (probabilistic)
{\bisimilarity}.

The remainder of the paper is organized as follows. In
Section~\ref{sec-prelim} we gather some notation regarding probability
distributions. For illustration purposes Section~\ref{sec-nondet}
treats the simpler setting of non-deterministic processes reiterating
the completeness proof for the equational theory of~\cite{GW96:jacm}
for rooted branching {\bisimilarity}. Next, after introducing
branching probabilistic {\bisimilarity} and some of its fundamental
properties in Sections~\ref{sec-bpb} and~\ref{fundamental},
respectively, in Section~\ref{sec-prob} we prove the main result,
viz.\ the completeness of an equational theory for rooted branching
probabilistic {\bisimilarity}, following the same lines set out in
Section~\ref{sec-nondet}. In Section~\ref{sec-concl} we wrap up and
make concluding remarks.



\section{Preliminaries}

\label{sec-prelim}

Let $\DistrX$ be the set of distributions over the set~$X$ of finite
support. The support of a distribution $\mu$ is denoted as $\mathit{spt}(\mu)$.
Each distribution $\mu \in \DistrX$ can be represented as $\mu
= \bigoplusiinI \: p_i * x_i$ when $\mu(x_i) = p_i$ for $i \in I$ and
$\sumiinI \: p_i = 1$. 
We assume $I$ to be a finite index set.
In concrete cases, when no confusion arises, the separator~$\ast$ is
omitted from the notation.
For convenience later, we do not require $x_i
\neq x_{i'}$ for $i \neq i'$ nor $p_i > 0$ for $i, i' \in I$.

We use $\delta(x)$ to denote the Dirac distribution for~$x \in X$.
For $\mu, \nu \in \DistrX$ and $r \in [0,1]$ we define $\mu \oplusr
\nu \in \DistrX$ by $(\mu \oplusr \nu)(x) = r \mycdot \mu(x) + (1{-}r)
\mycdot \nu(x)$. By definition $\mu \probc0 \nu = \nu$ and $\mu
\probc1 \nu = \mu$. 
For an index set~$I$, $p_i \in [0,1]$ and $\mu_i \in \DistrX$, we
define $\bigoplusiinI \: p_i * \mu_i \in \DistrX$ by $(\bigoplusiinI
\: p_i * \mu_i)(x) = \sumiinI \: p_i \mycdot \mu_i(x)$ for~$x \in
X$. For $\mu = \bigoplusiinI \: p_i * \mu_i$, $\nu = \bigoplusiinI \:
p_i * \nu_i$, and $r \in [0,1]$ it holds that $\mu \oplusr \nu =
\bigoplusiinI \: ( \mu_i \oplusr \nu_i )$.

For a binary relation $\calR \subseteq \DistrX \times \DistrX$ we
use~$\calR^\dagger$ to denote its symmetric closure.



\section{Completeness: the non-deterministic case}

\label{sec-nondet}

In this section we present 
an approach to prove completeness of an axiomatic theory for branching
{\bisimilarity} exploiting the notion of a concrete process in the setting
of a basic process language. In the remainder of the paper we
extend the approach to a process language involving probabilistic choice.

\blankline

\noindent
We assume to be given a set of actions~$\calA$ including the so-called
silent action~$\tau$. The process language we consider is called a
Minimal Process Language in~\cite{BBR10:cup}. It provides
inaction~$\bfzero$, a prefix construct for each action $a \in \calA$,
and non-deterministic choice.

\begin{definition}[Syntax]
  The class~$\calE$ of non-deterministic processes over~$\calA$, with
  typical element~$E$, is given by\vspace{-.5ex}
  \[
    E \bnfeq \bfzero \mid \alpha \pref E \mid E + E\vspace{-1.5ex}
  \]
  with actions~$\alpha$ from~$\calA$.
\pagebreak[3]
\end{definition}

\noindent
The process $\bfzero$ cannot perform any action, $\alpha \pref E$ can
perform action $\alpha$ and subsequently behave as $E$, and $E_1+E_2$
represents the choice in behavior between $E_1$ and~$E_2$.

For $E \in \calE$ we define its complexity~$c(E)$ by $c(\bfzero) = 0$,
$c(\alpha \pref E) = c(E) + 1$, and $c(E + F) = c(E) + c(F)$.

The behavior of processes in~$\calE$ is given by a structured
operational semantics going back to~\cite{HM80:icalp}.

\begin{definition}[Operational semantics]
  \label{def-nd-transition-relation}
  The transition relation ${\rightarrow} \subseteq \calE \times \calA
  \times \calE$ is
  given by
  \halflineup
  \begin{displaymath}
    \begin{array}{c}
      \sosrule{}{\alpha \pref E \arrow{\alpha} E}
      \: \textsc{\small (pref)}
     \medskip \\
      \sosrule{E_1 \arrow{\alpha} E_1}{E_1 + E_2 \arrow{\alpha} E_1}
      \: \textsc{\small (nd-choice\,1)}
      \qquad
      \sosrule{E_2 \arrow{\alpha} E_2}{E_1 + E_2 \arrow{\alpha} E_2}
      \: \textsc{\small (nd-choice\,2)}
   \end{array}
  \end{displaymath}
\end{definition}

\medskip

\noindent
We have auxiliary definitions and relations derived from the
transition relation of Definition~\ref{def-nd-transition-relation}.
A process $E' \in \calE$ is called a derivative of a process~$E \in
\calE$ iff $E_0, \ldots, E_n \in \calE$ and $\alpha_1, \ldots,
\alpha_n$ exist such that $E \equiv E_0$, $E_{i{-}1} \arrow{\alpha_i}
E_i$, and $E_n \equiv E'$.
We define $\der(E) = \lc E' \in \calE \mid \mbox{$E'$ derivative of~$E$}
\rc$.
Furthermore, for $E, E' \in \calE$ and $\alpha \in \calA$ we write $E
\alphahidearrow E'$ iff $E \alphaarrow E'$, or $\alpha = \tau$ and $E
= E'$. We use $\Arrow{}$ to denote the reflexive transitive closure of
$\arrow{(\tau)}$.

The definitions of strong and branching {\bisimilarity} for~$\calE$ are
standard and adapted from~\cite{GW96:jacm,Mil89:phi}.

\begin{definition}[Strong and branching {\bisimilarity}]
  \label{def-branching-bisimilar}
  \begin{itemize}
  \item [(a)] A symmetric relation $\calR \subseteq \calE \times
    \calE$ is called a \textit{strong bisimulation relation} iff for all
    $E, E', F \in \calE$ if $E \mopcalR F$ and $E
    \arrow{\alpha} E'$ then there is an $F' \in \calE$ such that
    \begin{displaymath}
      F \alphaarrow F' \ \text{and} \ 
      E' \mopcalR F'.
    \end{displaymath}
  \item [(b)] A symmetric relation $\calR \subseteq \calE \times
    \calE$ is called a \textit{branching bisimulation relation} iff for
    all $E, E', F \in \calE$ if $E \mopcalR F$ and $E
    \arrow{\alpha} E'$, then there are $\Fbar, F' \in \calE$ such that
    \begin{displaymath}
      F \Arrow{} \Fbar, \ 
      \Fbar \alphahidearrow F', \ 
      E \mopcalR \Fbar, \ \text{and} \ 
      E' \mopcalR F'.
    \end{displaymath}
  \item [(c)] Strong {\bisimilarity}, denoted by~${\bisim} \subseteq \calE
    \times \calE$, and branching {\bisimilarity}, written as~${\brbisim}
    \subseteq \calE \times \calE$, are defined as the largest strong
    bisimulation relation on~$\calE$ and the largest branching
    bisimulation relation on~$\calE$, respectively.
  \end{itemize}
\end{definition}

\noindent
Clearly, in view of the definitions, strong {\bisimilarity} between two
processes implies branching {\bisimilarity} between the two processes.

If for a transition $E \arrow{\tau} E'$ we have that $E \brbisim E'$,
the transition is called \textit{inert}. A process~$\Ebar$ is called
\textit{concrete} iff it has no inert transitions, i.e., if
$E'\in\der(\Ebar)$ and $E' \arrow{\tau} E''$, then~$E' \nbrbisim
E''$. We write $\calEconc = \lc \Ebar \in \calE \mid \mbox{$\Ebar$
  concrete} \rc$.

\blankline

\noindent
Next we introduce a restricted form of branching {\bisimilarity}, called
rooted branching {\bisimilarity}, instigated by the fact that branching
{\bisimilarity} itself is not a congruence for the choice operator. This
makes branching {\bisimilarity} unsuitable for equational reasoning where
it is natural to replace subterms by equivalent terms. Note that weak
{\bisimilarity} has the same problem~\cite{Mil89:phi}.

For example, we have for any process~$E$ that $E$ and~$\tau \pref E$
are branching bisimilar, but in the context of a non-deterministic
alternative they may not, i.e., it is not necessarily the case $E + F
\brbisim \tau \pref E + F$. More concretely, although $\bfzero
\brbisim \tau \pref \bfzero$, it does not hold that $\bfzero + b \pref
\bfzero \brbisim \tau \pref \bfzero + b \pref \bfzero$. The
$\tau$-move of $\tau \pref \bfzero + b \pref \bfzero$ to~$\bfzero$ has
no counterpart in $\bfzero + b \pref \bfzero$ because $\bfzero + b \pref
\bfzero \nbrbisim \bfzero$.

\begin{definition}
  \label{def-branching-bisimilarity}
  A symmetric $\calR \mathbin\subseteq \calE \times \calE$ is called a
  rooted branching bisimulation relation iff for all $E, F \in \calE$
  such that $E \mopcalR F$ it holds that if $E \arrow{\alpha} E'$
  for~$\alpha \in \calA$, $E' \in \calE$ then $F \arrow{\alpha} F'$
  and $E' \brbisim F'$ for some~$F' \in \calE$. Rooted branching
  bisimilarity, denoted by
  ${\rootedbrbisim} \subseteq \calE \times \calE$, is defined as the
  largest rooted branching bisimulation relation.
\end{definition}

\noindent
The definition of rooted branching bisimilarity boils down to calling
processes $E, F \in \calE$ rooted branching bisimilar,
notation~$E \rootedbrbisim F$, iff (i)~$E \arrow{\alpha} E'$ implies
$F \arrow{\alpha} F'$ and~$E' \brbisim F'$ for some $F' \in \calE$
and, vice versa, (ii)~$F \arrow{\alpha} F'$ implies
$E \arrow{\alpha} E'$ and~$E' \brbisim F'$ for some $E' \in
\calE$. The formulation of Definition~\ref{def-branching-bisimilarity}
for the nondeterministic processes of this section corresponds
directly to the definition of rooted branching \emph{probabilistic}
bisimulation that we will introduce in Section~\ref{sec-bpb}, see
Definition~\ref{rooted-bpb}. 
 
\blankline

\noindent
Direct from the definitions we see ${\bisim} \subseteq
{\rootedbrbisim} \subseteq {\brbisim}$. As implicitly announced we
have a congruence result for rooted branching {\bisimilarity}.

\begin{lemma}[\cite{GW96:jacm}]
  $\rootedbrbisim$ is a congruence on~$\calE$ for the operators $\pref$ and $+$.
\end{lemma}

\noindent
It is well-known that strong and branching {\bisimilarity} for~$\calE$
can be equationally characterized~\cite{BBR10:cup,GW96:jacm,Mil89:phi}.

\begin{definition}[Axiomatization of $\bisim$ and~$\rootedbrbisim$]
  The theory~$\AXd$ is given by the axioms A1 to~A4 listed in
  Table~\ref{table-axiomatization-of-branching-bisimilarity}.
  The theory~$\AXdb$ contains in addition the axiom~\hyperlink{B}{B}.
\end{definition}

\begin{table}
  \vspace{-4ex}
  \centering
  \def\arraystretch{1.1}
  \begin{tabular}{|@{\;}l@{\ \;}l|}
    \hline
    A1 & $E + F = F + E$ \rule{0pt}{12pt} \\
    A2 & $(E + F) + G = E + ( F + G)$ \\
    A3 & $E + E = E$ \\
    A4 & $E + \bfzero = E$ 
    \rule[-5pt]{0pt}{12pt}
    \\ \hline
    \hypertarget{B}{B} & $\alpha \pref ( \, F + \tau \pref ( E + F )
    \, ) = \alpha 
    \pref ( E + F )$
    \rule{0pt}{12pt}\rule[-5pt]{0pt}{12pt} \\ 
    \hline
  \end{tabular}
  \def\arraystretch{1.0}

  \medskip

  \caption{Axioms for strong and branching {\bisimilarity}}
  \label{table-axiomatization-of-branching-bisimilarity}
  \vspace{-4ex}
\end{table}

\noindent
If two processes are provably equal, they are rooted branching
bisimilar.

\begin{lemma} [Soundness]
  \label{lemma-soundness-base}
  For all $E,F \in \calE$, if $\AXdb \vdash E = F$ then $E
  \rootedbrbisim F$.
\end{lemma}

\begin{proof}[Sketch]
  First one shows that the left-hand side and the right-hand
  side of the axioms of~$\AXdb$ are rooted branching bisimilar.
  Next, one observes that rooted branching {\bisimilarity} is a
  congruence.
\end{proof}

\noindent
It is well-known that strong {\bisimilarity} is equationally
characterized by the axioms A1 to~A4 of
Table~\ref{table-axiomatization-of-branching-bisimilarity}.

\begin{theorem}[$\AXd$ sound and complete for~$\bisim$]
  \label{theorem-AXd-sound-and-complete}
  For all processes $E, F \in \calE$ it holds that $\AXd \vdash E =
  F$ iff $E \bisim F$.
\end{theorem}

\begin{proof}
  See for example~\cite[Section~7.4]{Mil89:phi}.
\end{proof}

\noindent
For concrete processes that have no inert transitions, branching
{\bisimilarity} and strong {\bisimilarity} coincide. Hence, in view of
Theorem~\ref{theorem-AXd-sound-and-complete}, branching {\bisimilarity}
implies equality for~$\AXd$.

\begin{lemma}
  \label{lemma-branching-implies-equal-for-concrete-processes}
  For all concrete $\Ebar,\Fbar \in \calEconc$, if $\Ebar \brbisim
  \Fbar$ then both $\Ebar \bisim \Fbar$ and $\AXd \vdash \Ebar = \Fbar$.
\end{lemma}

\begin{proof}[Sketch]
  Consider $\Ebar, \Fbar \in \calEconc$ such that $\Ebar \brbisim
  \Fbar$. Let~$\calR$ be a branching bisimulation relation relating
  $\Ebar$ and~$\Fbar$. Define $\calR'$ as the restriction of~$\calR$
  to the derivatives of $\Ebar$ and~$\Fbar$, i.e., $\calR' = \calR
  \cap ( (\der(\Ebar) \times \der(\Fbar)) \cup (\der(\Fbar) \times
  \der(\Ebar)) )$. Then $\calR'$ is a strong bisimulation relation,
  since none of the processes involved admits an inert
  $\tau$-transition. By the completeness of~$\AXd$, see
  Theorem~\ref{theorem-AXd-sound-and-complete}, it follows that $\AXd
  \vdash \Ebar = \Fbar$.
\end{proof}

\noindent
We are now in a position to prove the main technical result of this
section, viz.\ that branching {\bisimilarity} implies equality under a
prefix. In the proof the notion of a concrete process plays a central
role. 

\begin{lemma}
  \label{lemma-branching-bisim-vs-equal-under-prefix}
  \mbox{}
  \begin{itemize}
  \item [(a)] For all processes $E \in \calE$, a
    concrete process $\Ebar \in \calEconc$ exists such that $E \brbisim
    \Ebar$ and $\AXdb \vdash { \alpha \pref E = \alpha \pref \Ebar }$
    for all $\alpha \in \calA$.
  \item [(b)] For all processes $F,G \in \calE$, if $F \brbisim
    G$ then $\AXdb \vdash { \alpha \pref F = \alpha \pref G }$ for all
    $\alpha \in \calA$.%
  \pagebreak[3]
 \end{itemize}
\end{lemma}

\begin{proof}
  We prove statements (a) and~(b) by simultaneously induction on 
  $c(E)$ and $\max \lbrace c(F), c(G) \rbrace$, respectively.

  Basis, $c(E) = 0$. We have that $E= \bfzero+\cdots+\bfzero$. Hence,
  take $\Ebar=\bfzero$. Clearly, part (a) of the lemma holds as
  $\bfzero$ is concrete, $E \brbisim \bfzero$ and $\AXdb \vdash\alpha
  \pref E = \alpha \pref \bfzero$ for all $\alpha \in \calA$.

  Induction step for (a): $c(E) > 0$. The process $E$ can be written as 
  ${ \sumiinI \: \alpha_i \pref E_i }$
  for some finite $I$ and suitable~$\alpha_i\in\calA$ and $E_i\in\calE$.

  First suppose that for some $i_0\in I$ we have $\alpha_{i_0} = \tau$
  and $E_{i_0} \brbisim E$.  Then $\AXd \vdash E = H + \tau.E_{i_0}$,
  where \plat{$H := \sum_{i\in I\setminus\{i_o\}} \alpha_i \pref
    E_i$}.
  \vspace{1pt} By the induction hypothesis~(a), there is a term
  $\Ebar_{i_0} \in \calEconc$ such that $E_{i_0} \brbisim
  \Ebar_{i_0}$. We claim that $\Ebar_{i_0} \brbisim E_{i_0}+H$.

  For suppose $\Ebar_{i_0} \alphaarrow F$.
  \vspace*{1pt}
  Then $E_{i_0} \Arrow{} E'_{i_0} \alphahidearrow G$ where
  $\Ebar_{i_0}\brbisim E'_{i_0}$ 
  and $F \brbisim G$. In case $E_{i_0} = E'_{i_0}$ it follows that $E_{i_0}
  \alphahidearrow G$. 
  Since $\Ebar_{i_o}$ is concrete, either $\alpha \neq \tau$ or $F
  \nbrbisim \Ebar_{i_0}$. 
  Hence, $\alpha \neq \tau$ or $G \nbrbisim E_{i_0}$. So $E_{i_0}
  \alphaarrow G$. 
  Consequently, $E_{i_0}+H \alphaarrow G$.\linebreak[3]
  In case $E_{i_0} \neq E'_{i_0}$ we have $E_{i_0}+H \Arrow{} E'_{i_0}
  \alphahidearrow G$. 

  Now suppose $E_{i_0} + H \alphaarrow F$.  Then either
  $E_{i_0} \alphaarrow F$ or $H \alphaarrow F$.  \vspace*{1pt} In the
  first case we have
  $\Ebar_{i_0} \Arrow{} \Ebar'_{i_0} \alphahidearrow G$ where
  $E_{i_0}\brbisim \Ebar'_{i_0}$ and $F \brbisim G$, \vspace{1pt}
  while in the latter case $E \alphaarrow F$, and since
  $E \brbisim E_{i_0} \brbisim \Ebar_{i_0}$ we have
  $\Ebar_{i_0} \Arrow{} \Ebar'_{i_0} \alphahidearrow G$ where
  $E\brbisim \Ebar'_{i_0}$ and $F \brbisim G$. Because $\Ebar_{i_0}$
  is concrete, $\Ebar'_{i_o} = \Ebar_{i_0}$.  Thus
  $\Ebar_{i_0} \alphahidearrow G$ with $F \brbisim G$, which was
  to be shown.

  Hence $E_{i_0} \brbisim \Ebar_{i_0} \brbisim E_{i_0}+H$. Clearly
  $c(E_{i_0}), c(E_{i_0}+H) < c(E)$.  Therefore, by the induction
  hypothesis (b),
  $\AXdb\vdash \tau \pref E_{i_0} = \tau \pref ( E_{i_0} + H )$.  By
  the induction hypothesis (a), there is a term $\Ebar \in \calEconc$
  such that $\Ebar \brbisim E_{i_0} + H$ \vspace{1pt} and
  $\AXdb\vdash \alpha \pref \Ebar = \alpha \pref ( E_{i_0} + H )$. Now
  we have $E \brbisim E_{i_0} \brbisim E_{i_0} + H \brbisim
  \Ebar$. Therefore,
  \begin{align*}
    \AXdb\vdash\alpha \pref E 
    & = \alpha \pref ( H + \tau \pref E_{i_0} ) 
    && \\
    & = \alpha \pref ( H + \tau \pref ( E_{i_0} + H ) ) 
    && (\text{since $\AXdb\vdash\tau \pref E_{i_0} = \tau \pref
      (E_{i_0} + H)$}) \\ 
    & = \alpha \pref ( E_{i_0} + H ) 
    && (\text{use axiom \hyperlink{B}{B}}) \\
    & = \alpha \pref \Ebar
    && (\text{by the choice of $\Ebar$}).
  \end{align*}
  Hence, we have shown the existence of a desired process $\Ebar$ with
  the required properties.

  Now suppose, for all $i \in I$ we have $\alpha_i \neq \tau$ or
  $E_i\nbrbisim E$. Clearly $c(E_i) < c(E)$ for all $i \in I$.  By
  the induction hypothesis we can find, for all $i \in I$, concrete
  $\Ebar_i$ such that $\Ebar_i \brbisim E_i$ and $\AXdb \vdash \alpha
  \pref \Ebar_i = \alpha \pref E_i$ for all $\alpha \in \calA$. Define
  $\Ebar = { \sumiinI \: \alpha_i \pref \Ebar_i }$. Then $\Ebar
  \brbisim E$ and $\Ebar$~is concrete too, since $\Ebar_i \brbisim E_i
  \nbrbisim E \brbisim \Ebar$ for~$i \in I$ in case $\alpha_i = \tau$.
  Moreover, $\AXdb\vdash E = \Ebar$, since $E = \sumiinI \: \alpha_i
  \pref E_i = \sumiinI \: \alpha_i \pref \Ebar_i = \Ebar$. Hence, for
  $\alpha \in \calA$, $\AXdb \vdash \alpha \pref E = \alpha \pref
  \Ebar$.

  Both the base and the induction step for (b): $\max \lbrace c(F),
  c(G) \rbrace \geqslant 0$. Suppose $F \brbisim G$. Pick $\Fbar,
  \Gbar \in \calEconc$ such that $F \brbisim \Fbar$ and
  $\AXdb\vdash\alpha \pref F = \alpha \pref \Fbar$ for all~$\alpha \in
  \calA$, and similarly for $G$ and~$\Gbar$. Then we have $\Fbar
  \brbisim \Gbar$. Since $\Fbar$ and~$\Gbar$ are concrete it follows
  that $\AXd\vdash\Fbar = \Gbar$, see
  Lemma~\ref{lemma-branching-implies-equal-for-concrete-processes}. Now
  pick any $\alpha \in \calA$. Then we have $\AXdb\vdash\alpha \pref F
  = \alpha \pref \Fbar = \alpha \pref \Gbar = \alpha \pref G$.
\end{proof}

\noindent
By now we have gathered sufficient building blocks to prove the main
result of this section.

\begin{theorem}[$\AXdb$ sound and complete for $\rootedbrbisim$]
  \label{theorem-completeness}
  For all processes $E, F \in \calE$ it holds that $E \rootedbrbisim
  F$ iff $\AXdb \vdash E = F$.
\end{theorem}

\begin{proof}
  In view of Lemma~\ref{lemma-soundness-base} we only need to prove
  completeness of~$\AXdb$ for rooted branching
  {\bisimilarity}. Suppose $E, F \in \calE$ and $E \rootedbrbisim
  F$. Let $E = \sumiinI \: \alpha_i \pref E_i$ and $F= \sumjinJ \:
  \beta_j \pref F_j$ for suitable index sets $I$ and~$J$, $\alpha_i,
  \beta_j \in \calA$, $E_i, F_j \in \calE$. Since $E \rootedbrbisim F$
  we have (i)~for all $i \in I$ there is a $j \in J$ such that
  $\alpha_i = \beta_j $ and $E_i \brbisim F_j$, and, symmetrically,
  (ii)~for all $j \in J$ there is an $i \in I$ such that $\alpha_i =
  \beta_j$ and $E_i \brbisim F_j$.  Put $K = \lc (i,j) \in I \times J
  \mid ( \alpha_i = \beta_j ) \land ( E_i \brbisim F_j ) \rc$. Define
  the processes $G, H \in \calE$ by
  \begin{displaymath}
    G = \sumkinK \: \gamma_k \pref G_k
    \quad \text{and} \quad
    H = \sumkinK \: \zeta_k \pref H_k
  \end{displaymath}
  where, for $i \in I$, $\gamma_k = \alpha_i$ and $G_k \equiv E_i$ if
  $k = (i,j)$ for some $j \in J$, and, similarly for $j \in J$,
  $\zeta_k = \beta_j$ and $H_k \equiv F_j$ if $k = (i,j)$ for some $i
  \in I$. Then $G$ and~$H$ are well-defined. Moreover, $\AXd \vdash E
  = G$ and $\AXd \vdash F = H$.

  For $k \in K$, say $k = (i,j)$, it holds that $\gamma_k = \alpha_i =
  \beta_j = \zeta_k$ and $G_k \equiv E_i \brbisim F_j \equiv H_k$, by
  definition of~$K$. By
  Lemma~\ref{lemma-branching-bisim-vs-equal-under-prefix}b we obtain,
  for all $k \in K$, $\AXdb \vdash { \gamma_k \pref G_k = \zeta_k
    \pref H_k }$. From this we get
  \begin{displaymath}
    \AXdb \vdash E 
    = \sumiinI \: \alpha_i \pref E_i
    = \sumkinK \: \gamma_k \pref G_k
    = \sumkinK \: \zeta_k \pref H_k
    = \sumjinJ \: \beta_j \pref F_j
    = F
  \end{displaymath}
  which concludes the proof of the theorem.
\end{proof}



\section{Branching {\bisimilarity} for probabilistic processes} 
\label{sec-bpb}

In this section we define branching {\bisimilarity}
for probabilistic processes.

Following~\cite{BS01:icalp}, we start with adapting the syntax of
processes, now distinguishing non-deterministic processes~$E \in
\calE$ and probabilistic processes~$P \in \calP$.

\begin{definition}[Syntax]
  The classes $\calE$ and~$\calP$ of non-deterministic and
  probabilistic processes over~$\calA$, respectively, ranged over by
  $E$ and~$P$, are given by
  \begin{align*}
    E & \bnfeq \bfzero \mid \alpha \pref P \mid E + E \\
    P & \bnfeq \partial(E) \mid P \oplusr P 
  \end{align*}
  with actions~$\alpha$ from~$\calA$ where $r \in (0,1)$.
\end{definition}

\noindent
The probabilistic process $P_1 \oplusr P_2$ executes the behavior
of~$P_1$ with probability~$r$ and the behavior~$P_2$ with
probability~$1-r$. By convention, $P \probc1 Q$ denotes~$P$ and $P
\probc0 Q$ denotes~$Q$.

We again introduce the complexity measure~$c$, now for
non-deterministic and probabilistic processes, based on the depth of a
process.  The complexity measure $c : \calE\cup \calP \to \bbbn$
is given by $c(\bfzero) = 0$, $c(\alpha \pref P) = c(P) + 1$, $c(E+F)
= c(E)+c(F)$, and~$c(\partial(E)) = c(E) + 1$, $c(P \oplusr Q) = c(P) +
c(Q)$.

\blankline

\noindent
As usual SOS semantics for $\calE$ and~$\calP$ makes use of two types
of transition relations~\cite{HJ90:rtss,BS01:icalp}.

\pagebreak[3]

\begin{definition}[Operational semantics]
  \label{def-pr-transition-relation} \mbox{}
  \begin{itemize}
    \item [(a)] The transition relations ${\rightarrow} \subseteq
      \calE \times \calA \times \Distr(\calE)$ and ${\mapsto}
      \subseteq \calP \times \Distr(\calE)$ are given by\vspace{-2.5ex}
      \begin{displaymath}
        \begin{array}{c}
          \sosrule{P \mapsto \mu}{\alpha \pref P \arrow{\alpha} \mu}
          \: \textsc{\small (pref)}
          \medskip \\
          \sosrule{E_1 \arrow{\alpha} \mu_1}{E_1 + E_2 \arrow{\alpha}
            \mu_1}
          \: \textsc{\small (nd-choice\,1)}
          \qquad
          \sosrule{E_2 \arrow{\alpha} \mu_2}{E_1 + E_2 \arrow{\alpha}
            \mu_2}
          \: \textsc{\small (nd-choice\,2)}
          \bigskip \\
          \sosrule{}{\partial(E) \mapsto \delta(E)}
          \: \textsc{\small (Dirac)}
          \qquad
          \sosrule{P_1 \mapsto \mu_1 \quad P_2 \mapsto
            \mu_2}{P_1 \oplusr P_2 \mapsto \mu_1 \oplusr \mu_2}
          \: \textsc{\small (p-choice)}
        \end{array}
      \end{displaymath}
    \item [(b)] The transition relation ${\rightarrow}
      \subseteq \DistrE \times \calA \times \DistrE$ is such that $\mu
      \alphaarrow \mu'$ whenever $\mu = \bigoplusiinI \: p_i * E_i$,
      $\mu' = \bigoplusiinI \: p_i * \mu'_i $, and $E_i \alphaarrow
      \mu'_i$ for all $i \in I$.
      \vspace{-.5ex}
  \end{itemize}
\end{definition}

\noindent
With $\den{P}$, for $P\in\calP$, we denote the unique distribution $\mu$ such that $P\mapsto\mu$.

The transition relation~$\rightarrow$ on distributions allows for a
probabilistic combination of non-deterministic alternatives resulting
in a so-called combined transition,
cf.~\cite{SL94:concur,Seg95:thesis}. For example, for $E \equiv a
\pref (P \probc{1/2} Q) + a \pref (P \probc{1/3} Q)$, the Dirac
process $\delta(E) \equiv \delta( a \pref (P \probc{1/2} Q) + a \pref
(P \probc{1/3} Q) )$ provides an $a$-transition to $\den{P \probc{1/2} Q}$
as well as an $a$-transition to $\den{P \probc{1/3} Q}$. However, since for
distribution~$\delta(E)$ it holds that $\delta(E) = \frac12 \delta(E)
\oplus \frac12 \delta(E)$ there is also a transition
\begin{displaymath}
  \delta(E) 
  = \half \delta(E) \oplus \half \delta(E)
  \arrow{a} \half \den{P \probc{1/2\mkern1mu} Q} \oplus \half \den{P
  \probc{1/3\mkern1mu} Q} = \den{P \probc{5/12\mkern2mu} Q}.
\end{displaymath}
As noted in~\cite{Sto02:phd}, the ability to combine transitions is
crucial for obtaining transitivity of probabilistic process
equivalences that take internal actions into account.

Referring to the example in the introduction, the processes of $t_0$
and~$u_0$ will be identified.  However, without the splitting of the
source distribution~$\mu$ as provided by
Definition~\ref{def-pr-transition-relation}, we are not able to relate
$t_0$ and~$u_0$ directly, or rather their direct derivatives, while
meeting the natural transfer conditions (see
Definition~\ref{def-probabilistically-branching-bisimilar}).  The
difficulty arises when both $P$ and $Q$ can do a $\tau$-transition to
non-bisimilar processes.

In preparation to the definition of the notion of branching
probabilistic {\bisimilarity} below we introduce some notation.

\begin{definition}
  For $\mu, \mu' \mathbin\in \DistrE$ and $\alpha \mathbin\in \calA$ we write
  $\mu \alphahidearrow \mu'$ iff (i)~$\mu \alphaarrow \mu'$, or
  (ii)~$\alpha = \tau$ and $\mu = \mu_1 \oplusr \mu_2$,
  $\mu' = \mu'_1 \oplusr \mu'_2$ such that $\mu_1 \arrow{\tau} \mu_1'$
  and $\mu_2 = \mu'_2$ for some $r \in [0,1]$. We use
  $\Arrow{}$ to denote the reflexive transitive closure of
  $\arrow{(\tau)}$.
\end{definition}

\noindent
Thus, for example,\vspace{-1ex}
\begin{gather*}
  \myfrac13 \delta( \tau \pref (P \probc{1/2} Q)) \oplus
  \myfrac23 \den{ P \probc{1/2} Q } \  \arrow{(\tau)} \   \den{P \probc{1/2} Q} 
  \quad \text{and} \\
  \myfrac12 \delta( \tau \pref \partial ( \tau \pref P )) \oplus
  \myfrac13 \delta( \tau \pref P ) \oplus \myfrac16 \den{P}
  \Arrow{} \den{P}.
\end{gather*}

We are now in a position to define strong probabilistic {\bisimilarity}
and branching probabilistic {\bisimilarity}. Note that the notion of strong
probabilistic {\bisimilarity} is the variant with combined transitions as defined in 
\cite{SL94:concur,BS01:icalp}. 

\begin{definition}[Strong and branching probabilistic {\bisimilarity}]
  \label{def-probabilistically-branching-bisimilar}
  \begin{itemize}
  \item [(a)] A symmetric relation $\calR \subseteq \DistrE \times
    \DistrE$ is called \textit{decomposable} iff for all $\mu, \nu \in
    \DistrE$ such that $\mu \mopcalR \nu$ and $\mu = \bigoplusiinI \:
    p_i * \mu_i$ there are $\nu_i \in \DistrE$, for~$i \in I$,
    such that
    \begin{displaymath}
      \nu  = \bigoplusiinI \: p_i * \nu_i \ 
      \text{and} \ 
      \mu_i \mopcalR \nu_i 
      \ \text{for all~$i \in I$.}
    \end{displaymath}

  \item [(b)] A decomposable relation $\calR \subseteq \DistrE \times
    \DistrE$ is called a \emph{strong probabilistic bisimulation
    relation} iff for all $\mu, \nu \in \DistrE$ such that $\mu
    \mopcalR \nu$ and $\mu \arrow{\alpha} \mu'$ there is a $\nu'
    \in \DistrE$ such that
    \begin{displaymath}
      \nu \alphaarrow \nu' \ \text{and} \ 
      \mu' \mopcalR \nu'.
    \end{displaymath}

  \item [(c)] A symmetric relation $\calR \subseteq \DistrE \times
    \DistrE$ is called \textit{weakly decomposable} iff for all $\mu,
    \nu \in \DistrE$ such that $\mu \mopcalR \nu$ and $\mu =
    \bigoplusiinI \: p_i * \mu_i$ there are $\nubar, \nu_i \in
    \DistrE$, for~$i \in I$, such that
    \begin{displaymath}
      \nu \Arrow{} \nubar,\ 
      \mu \mopcalR \nubar,\  
      \nubar = \bigoplusiinI \: p_i * \nu_i,\ 
      \text{and} \ 
      \mu_i \mopcalR \nu_i 
      \ \text{for all~$i \in I$.}
    \end{displaymath}

  \item [(d)] A weakly decomposable relation $\calR \subseteq \DistrE
    \times \DistrE$ is called a \emph{branching} probabilistic
    bisimulation relation iff for all $\mu, \nu \in \DistrE$ such
    that $\mu \mopcalR \nu$ and $\mu \arrow{\alpha} \mu'$, there are
    $\nubar, \nu' \in \DistrE$ such that
    \begin{displaymath}
      \nu \Arrow{} \nubar, \ 
      \nubar \alphahidearrow \nu', \ 
      \mu \mopcalR \nubar, \ \text{and} \ 
      \mu' \mopcalR \nu'.
    \end{displaymath}

  \item [(e)] Strong probabilistic {\bisimilarity}, denoted by~${\bisim}
    \subseteq \DistrE \times \DistrE$, and branching probabilistic
    {\bisimilarity}, written as~${\brbisim} \subseteq \DistrE \times
    \DistrE$, are respectively defined as the largest strong
    probabilistic bisimulation relation on~$\DistrE$ and as the
    largest branching probabilistic bisimulation relation
    on~$\DistrE$.
  \end{itemize}
\end{definition}

\noindent
By comparison, on finite processes, as used in this paper,
the branching probabilistic bisimilarity of Segala \&
Lynch~\cite{SL94:concur} can be defined in our framework exactly as in
(d) and~(e) above, but taking a decomposable instead of a weakly
decomposable relation. This yields a strictly finer equivalence,
distinguishing the processes $s_0$, $t_0$ and $u_0$ from the
introduction.

The notion of decomposability has been adopted from~\cite{Hen12:facj}
and weak decomposability from~\cite{LV16}. The underlying idea stems
from \cite{DGHM09}. These notions provide a convenient dexterity to
deal with behavior of sub-distributions, e.g., to distinguish
$\myfrac12 \partial( a \pref \partial(\bfzero)) \oplus \myfrac12
\partial( b \pref \partial(\bfzero))$ from $\partial(\bfzero)$, as
well as combined behavior.

\newcommand*\circled[1]{\tikz[baseline=(char.base)]{
    \node[shape=circle,draw,inner sep=0.5pt] (char) {#1};}}

Our definition of branching probabilistic {\bisimilarity} is based on
distributions rather than on states and has similarity with the notion
of weak distribution {\bisimilarity} proposed by Eisentraut et
al.\ in~\cite{EHKTZ13:qest}. Consider the running example
of~\cite{EHKTZ13:qest}, reproduced in Figure~\ref{PA-of-Eisentraut}
and reformulated in terms of the process language at hand. The states
\circled{1} and~\circled{6} are identified with respect to weak
distribution {\bisimilarity} as detailed
in~\cite{EHKTZ13:qest}. Correspondingly, putting
\begin{displaymath}
  \begin{array}{rcl}
  E_1 & = & \tau \pref \bigl(
  \partial\dg( 
    \tau \pref \partial\dor( \tau \pref P + c \pref Q + \tau \pref R
    \dor) + 
    c \pref Q + \tau \pref R \dg) 
  \probc{1/2}
  {} \\ & & \phantom{\tau \pref \bigl( {}}
  \partial\dg( \tau \pref \mg(
    \partial\dor( \tau \pref P + c \pref Q + \tau \pref R \dor)
    \probc{1/2}
    \partial( \bfzero ) \mg) \dg)
  \bigr) \smallskip \\  
  E_6 & = & \tau \pref ( \partial\dor( \tau \pref P + c \pref Q + \tau \pref
  R \dor) \probc{3/4} \partial( \bfzero ) )
  \end{array}
\end{displaymath}
the non-deterministic processes
$E_1$ and~$E_6$ are identified with respect to branching probabilistic
{\bisimilarity}.

\begin{figure}

\begin{center}
\begin{tikzpicture}[> = stealth', semithick, inner sep = 1pt, scale=0.90]
  \tikzstyle{state} = [circle, draw=Black, fill=LightYellow, minimum size=4.0mm]
  \tikzstyle{split} = [circle, draw=black, fill=black]
  \node [state] (s1) at (0,1) {$1$}; 
  \node [split] (x1) at (0.75,1) {};
  \node [state] (s2) at (1.5,2) {$2$}; 
  \node [split] (x2a) at (2.25,2.5) {}; 
  \node [split] (x2b) at (1.875,3) {}; 
  \node [split] (x2) at (3.0,2) {}; 
  \node [state] (s3) at (1.5,0) {$3$}; 
  \node [split] (x3) at (3.0,0) {}; 
  \node [state] (s4) at (4.5,2) {$4$}; 
  \node [split] (x4) at (3.75,2.5) {}; 
  \node [split] (x4a) at (4.15,3) {}; 
  \node [split] (x4b) at (4.85,3) {}; 
  \node [state] (s5) at (4.5,0) {$5$}; 
  \node [state] (s6) at (6,1) {$6$}; 
  \node [split] (x6) at (5.25,1) {}; 

  \node [inner sep=0pt] (Q) at (3,3) {\large \textcolor{green}{\trianglepafill}}; 
  \node [inner sep=0pt] (P) at (3,4) {\large \textcolor{blue}{\squadfill}}; 
  \node [inner sep=1pt] (R) at (4.5,4) {\large \textcolor{red}{\pentagofill}}; 
  \path [>=] (s1) edge node [above=1pt] {$\tau$} (x1);
  \path [->] (x1) edge node [midway, above left] {$\frac{1}{2}$} (s2);
  \path [->] (x1) edge node [midway, below left] {$\frac{1}{2}$} (s3);
  \draw [>=] (x1) ++ (0.10,-0.15) arc (-48.0:48:0.2);
  \path [>=] (s2) edge node [above=1pt] {$\tau$} (x2);
  \path [->] (x2) edge node [midway, above=1pt] {$1$} (s4);
  \draw [>=] (x2) ++ (0.10,-0.15) arc (-48.0:48:0.2);  
  \path [>=] (s2) edge node [above left=0pt] {$c$} (x2a);
  \path [->] (x2a) edge node [midway, above left=1pt] {$1$} (Q);
  \draw [>=] (x2a) ++ (0.15,0.00) arc (10:70:0.2);
  \path [>=] (s2) edge node [above left=0pt] {$\tau$} (x2b);
  \path [->] (x2b) edge node [midway, above left=1pt] {$1$} (P);
  \draw [>=] (x2b) ++ (0.15,0.00) arc (10:70:0.2);
  \path [>=] (s3) edge node [above=1pt] {$\tau$} (x3);
  \path [->] (x3) edge node [midway, above left] {$\frac{1}{2}$} (s4);
  \path [->] (x3) edge node [midway, below=1pt] {$\frac{1}{2}$} (s5);
  \draw [>=] (x3) ++ (0.18,0.00) arc (0.0:50:0.2);
  \path [>=] (s4) edge node [above right=0pt] {$c$} (x4);
  \path [->] (x4) edge node [above right] {$1$} (Q);
  \draw [>=] (x4) ++ (-0.05,+0.15) arc (120:180:0.2);
  \path [>=] (s4) edge node [above right=0pt] {$\tau$} (x4a);
  \path [->] (x4a) edge node [above right] {$1$} (P);
  \draw [>=] (x4a) ++ (-0.05,+0.15) arc (110:170:0.2);
  \path [>=] (s4) edge node [right=1pt] {$\tau$} (x4b);
  \path [->] (x4b) edge node [right=1pt] {$1$} (R);
  \draw [>=] (x4b) ++ (0.04,+0.18) arc (90:140:0.2);
  \path [>=] (s6) edge node [above=1pt] {$\tau$} (x6);
  \path [->] (x6) edge node [midway, above right] {$\frac{3}{4}$} (s4);
  \path [->] (x6) edge node [midway, below right] {$\frac{1}{4}$} (s5);
  \draw [>=] (x6) ++ (-0.10,+0.15) arc (132.0:228:0.2);
\end{tikzpicture}
\end{center}
  \vspace*{-0.50cm}
  \caption{Probabilistic automaton of~\cite{EHKTZ13:qest}}
  \label{PA-of-Eisentraut}
\end{figure}
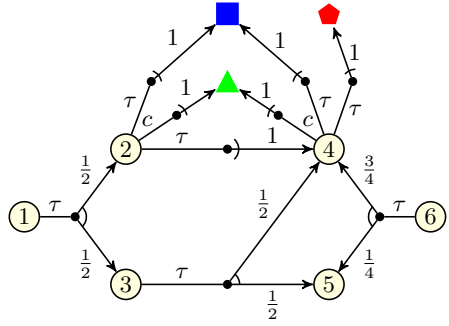

Note that strong and branching probabilistic {\bisimilarity} are
well-defined since any union of strong or branching probabilistic
bisimulation relations is again a strong or branching probabilistic
bisimulation relation. In particular, (weak) decomposability is
preserved under arbitrary unions.

\blankline

\noindent
As we did for the non-deterministic setting, we introduce a notion of
rooted branching probabilistic {\bisimilarity} for distributions over
processes. 

\begin{definition}
  \label{rooted-bpb}
  A symmetric and decomposable relation $\calR \mathbin\subseteq \DistrE
  \times \DistrE$ is called a rooted branching probabilistic
  bisimulation relation iff for all $\mu, \nu \in \DistrE$ such that
  $\mu \mopcalR \nu$ it holds that if $\mu \arrow{\alpha} \mu'$
  for~$\alpha \in \calA$, $\mu' \in \DistrE$ then $\nu \arrow{\alpha}
  \nu'$ and $\mu' \brbisim \nu'$ for some~$\nu' \in \DistrE$. Rooted
  branching probabilistic bisimilarity, denoted by ${\rootedbrbisim}
  \subseteq \DistrE \times \DistrE$, is defined as the largest rooted
  branching probabilistic bisimulation relation.
\end{definition}

\noindent
  Since any union of rooted branching probabilistic bisimulation
  relations is again a rooted branching probabilistic bisimulation
  relation, rooted branching probabilistic
  {\bisimilarity}~$\rootedbrbisim$ is well-defined.

  Note, the two probabilistic processes
  \begin{displaymath}
    P = \partial( \tau \pref \partial(\astop)) \probc{1/2} \partial(\bstop)
    \quad \text{and} \quad
    Q = \partial(\astop) \probc{1/2} \partial(\bstop)
  \end{displaymath}
  are \emph{not} rooted branching probabilistically bisimilar. Any
  rooted branching probabilistic bisimulation relation is by
  decomposability required to relate the respective probabilistic
  components $\partial( \tau \pref \partial(\astop))$ and
  $\partial(\astop)$, which clearly do not meet the transfer
  condition. Thus, since $\partial( \tau \pref \partial(\astop))
  \nrootedbrbisim \partial(\astop)$ also~$P \nrootedbrbisim Q$.

Two non-deterministic processes are considered to be strongly, rooted
branching, or branching probabilistically bisimilar iff their Dirac
distributions are, i.e., $E \bisim F$ iff $\delta(E) \bisim \delta(F)$,
$E \rootedbrbisim F$ iff $\delta(E) \rootedbrbisim \delta(F)$, and $E
\brbisim F$ iff $\delta(E) \brbisim \delta(F)$. 
Two probabilistic processes are considered to be strongly, rooted
branching, or branching probabilistically bisimilar iff their
associated distributions over~$\calE$ are.

We show that branching probabilistic bisimilarity, although not a
congruence for non-deterministic choice, is a congruence for
probabilistic choice. We first need a technical result.

  \begin{lemma}
    \label{flexibelDelen}
    Let $I$ and~$J$ be finite index sets, $p_i,q_j\in[0,1]$ and
    $\xi, \mu_i,\nu_j \in \DistrE$, for $i \in I$ and $j \in J$, with
    $\xi = \bigoplusiinI\: p_i*\mu_i$ and
    $\xi =\bigoplusjinJ \: q_j * \nu_j$. Then there are
    $r_{ij} \in [0,1]$ and $\varrho_{ij} \in \DistrE$ such that
    $\sum_{i\in I} \: r_{ij}=q_j$, $\sum_{j\in J} \: r_{ij}=p_i$,
    \begin{math}
      p_i * \mu_i = \bigoplusjinJ \: r_{ij} * \varrho_{ij}
      \textrm{ for all}~i \in I,
      \textrm{ and } 
      q_j * \nu_j = \bigoplusiinI \: r_{ij} * \varrho_{ij}
      \textrm{ for all}~j \in J.
    \end{math}
\end{lemma}

\begin{proof}
  Let
  \begin{math}\displaystyle
    r_{ij} = \sum_{E\in \mathit{spt}(\xi)} \:
    \frac{\rule[-3pt]{0pt}{10pt} p_i\mu_i(E) \, q_j
      \nu_j(E)}{\rule{0pt}{7pt} \xi(E)}
  \end{math}
  for all $i \in I$ and $j \in J$.
  In case $r_{ij}=0$ choose $\varrho_{ij}\in\DistrE$ arbitrarily.
  Otherwise, define for all $E \in \cal{E}$, $i \in I$, $j \in J$:
  \begin{displaymath}
    \varrho_{ij}(E) =
    \left\lbrace
      \begin{array}{cl}
        \displaystyle\frac
        {\rule[-3pt]{0pt}{10pt} p_i\mu_i(E) \, q_j \nu_j(E)}
        {\rule{0pt}{7pt} r_{ij} \, \xi(E)} \quad &
        \textrm{if } \xi(E)>0, \\
        0 & \textrm{otherwise.}
      \end{array}
    \right.
  \end{displaymath}
  With these definitions it is straightforward to check the
  required properties.
\end{proof}

\begin{lemma}
  \label{unrooted congruence}
  Let $\mu_1,\mu_2,\nu_1,\nu_2\in \DistrE$ and $r\in(0,1)$. If $\mu_1
  \brbisim \nu_1$ and $\mu_2 \brbisim \nu_2$ then $\mu_1 \probc{r}
  \mu_2 \brbisim \nu_1 \probc{r} \nu_2$.
\end{lemma}

\begin{proof}
  Suppose $\mu_1 \brbisim \nu_1$ and $\mu_2 \brbisim \nu_2$ through
  branching probabilistic bisimulation relations $\calR_1$
  and~$\calR_2$. We show that the relation $\calR = \lc \langle \xi'
  \probc{s} \xi'' , \eta' \probc{s} \eta'' \rangle \mid \xi'
  \mathop{\calR_1} \eta' ,\, \xi'' \mathop{\calR_2} \eta'' ,\, {s \in
    (0,1)} \rc$ is a branching probabilistic bisimulation relation
  relating $\mu_1 \probc{r} \mu_2$ with $\nu_1 \probc{r} \nu_2$.
  
  Symmetry is straightforward. We show that $\calR$ is weakly
  decomposable. So, assume $\xi \mopcalR \eta$ and $\xi =
  \bigoplusiinI \: p_i * \xi_i$. Thus, $\xi = \xi' \probc{s} \xi''$
  and $\eta = \eta' \probc{s} \eta''$ with $\xi' \mathop{\calR_1}
  \eta'$ and $\xi'' \mathop{\calR_2} \eta''$ for suitable $\xi',
  \xi'', \eta', \eta'' \in \DistrE$. By Lemma~\ref{flexibelDelen}
  there must be $s'_i, s''_i \geqslant 0$ and $\xi'_i, \xi''_i \in
  \DistrE$, for $i \in I$, such that
  \begin{displaymath}
    \begin{array}{lcl}
      \xi' & = & \bigoplusiinI \: s'_i/s * \xi'_i \, , 
      \smallskip \\ 
      \xi''  & = & \bigoplusiinI \: s''_i/(1{-}s) * \xi''_i
      \, , 
      \smallskip \\ 
      \xi_i & = & ( s'_i/p_i * \xi'_i ) \oplus ( s''_i/p_i * \xi''_i )
      \quad \text{for all $i\in I$\,,}
    \end{array}
  \end{displaymath}
  $\sum_{i\in I}s'_i = s$, $\sum_{i\in I}s''_i = 1{-}s$,
  and $s'_i + s''_i = p_i$ for all $i \in I$.
  Since $\calR_1$ and~$\calR_2$ are weakly decomposable, there are
  $\etabar', \etabar'', \eta'_i, \eta''_i \in \DistrE$ for~$i \in I$
  such that
  \begin{displaymath}
    \begin{array}{lclclcl}
      \eta' \Arrow{} \etabar' & \quad &
      \xi' \mathop{\calR_1} \etabar' , & \quad &
      \etabar' \, = \bigoplusiinI \:
      {s'_i}/{s} * \eta'_i , & \quad &
      \xi'_i \mathop{\calR_1} \eta'_i
      , \smallskip \\ 
      \eta'' \Arrow{} \etabar'' &&
      \xi'' \mathop{\calR_2} \etabar'' &&
      \etabar'' = \bigoplusiinI \:
      {s''_i}/{(1{-}s)} * \eta''_i &&
      \xi''_i \mathop{\calR_1} \eta''_i 
  \end{array}
  \end{displaymath}
  for all $i \in I$. Therefore, we can conclude that 
  \begin{displaymath}
    \eta' \probc{s} \eta'' \Arrow{} \etabar' \probc{s} \etabar''
    \quad \text{and} \quad
    ( \xi' \probc{s} \xi'' ) \mopcalR \mkern1mu ( \etabar'
    \probc{s} \etabar'' )
    \, .
  \end{displaymath}
  Moreover, 
  \begin{displaymath}
    \begin{array}{rcl}
      \etabar' \probc{s} \etabar'' & = &
      \bigl( \bigoplusiinI \: {s'_i}/{s} * \eta'_i \bigr)
      \probc{s} 
      \bigl( \bigoplusiinI \: {s''_i}/{(1{-}s)} * \eta''_i \bigr) 
      \smallskip \\
      & = &
      \bigl( \bigoplusiinI \: s'_i * \eta'_i \bigr)
      \oplus 
      \bigl( \bigoplusiinI \: s''_i * \eta''_i \bigl)
      \smallskip \\ & = & \phantom{\bigl(}
      \bigoplusiinI \: p_i * ( \eta'_i \probc{s'_i/p_i} \eta''_i )
    \end{array}
  \end{displaymath}
  and $\xi_i = ( \xi'_i \probc{s'_i/p_i} \xi''_i ) \mopcalR
  \mkern1mu ( \eta'_i \probc{s'_i/p_i} \eta''_i )$ for all $i \in I$.
  This finishes the argument that $\calR$ is decomposable.
  
  Next, we show that $\calR$ satisfies the transfer property for
  bisimulations. Suppose we have $(\xi_1 \probc{r} \xi_2) \calR
  (\eta_1 \probc{r} \eta_2)$, thus $\xi_1 \calR_1 \xi_2$ and $\eta_1
  \calR_2 \eta_2$. If $\xi_1 \probc{r} \xi_2 \alphaarrow \xi'$ then
  $\xi_1 \alphaarrow \xi'_1$, $\xi_2 \alphaarrow \xi'_2$ and $\xi' =
  \xi'_1 \probc{r} \xi'_2$ for suitable $\xi'_1, \xi'_2 \in
  \DistrE$.\vspace{1pt} By assumption, $\bar\eta_1,\eta'_1$
  and~$\bar\eta_2, \eta'_2$ exist such that $\eta_1 \Arrow{} \etabar_1
  \alphahidearrow \eta'_1$, $\eta_2 \Arrow{} \etabar_2 \alphahidearrow
  \eta'_2$, $\xi_1 \mathop{\calR_1} \etabar_1$, $\xi_2
  \mathop{\calR_2} \etabar_2$, $\xi'_1 \mathop{\calR_1} \eta'_1$,
  and~$\xi'_2 \mathop{\calR_2} \eta'_2$. From this we obtain $\eta_1
  \probc{r} \eta_2 \Arrow{} \etabar_1 \probc{r} \etabar_2\linebreak[3]
  \alphahidearrow \eta'$, $\xi_1 \probc{r} \xi_2 \mopcalR
  \etabar_1 \probc{r} \etabar_2$ and~$\xi' \mopcalR \eta'$ for
  $\eta' = \eta'_1 \probc{r} \eta'_2$.
\end{proof}

\noindent
A direct consequence of the previous lemma is that if $P_1 \brbisim
Q_1$ and $P_2 \brbisim Q_2$ then $P_1 \probc{r} P_2 \brbisim Q_1
\probc{r} Q_2$.

\begin{lemma} [Congruence]
  \label{lemma-congruence-prob}
  The relations $\bisim$, $\rootedbrbisim$, and~$\brbisim$ on $\calE$
  and~$\calP$ are equivalence relations, and the relations $\bisim$ and
  $\rootedbrbisim$ are congruences on $\calE$ and~$\calP$.
\end{lemma}

\begin{proof}
  The proof of $\bisim$, $\rootedbrbisim$, and~$\brbisim$ being
  equivalence relations involves a number of straightforward auxiliary
  results, in particular for the case of transitivity, and are omitted
  here.

  Regarding congruence the interesting cases are for non-deterministic
  and probabilistic choice with respect to rooted branching probabilistic
  {\bisimilarity}.
  Suppose $E_1 \rootedbrbisim F_1$ and $E_2 \rootedbrbisim F_2$. Then
  $\calR = \singleton{ \langle \delta(E_1+E_2) , \delta(F_1+F_2)
    \rangle }^\dagger$ is a rooted branching probabilistic
  bisimulation relation. Clearly, $\calR$~is symmetric and
  decomposable. Moreover, if $\delta(E_1+E_2) \arrow{\alpha} \mu'$,
  then either $\delta(E_1) \alphaarrow \mu'$, $\delta(E_2) \alphaarrow
  \mu'$, or $\delta(E_1) \alphaarrow \mu'_1$, $\delta(E_2) \alphaarrow
  \mu'_2$ and $\mu' = \mu'_1 \probc{r} \mu'_2$ for suitable $\mu'_1,
  \mu'_2 \in \DistrE$ and~$r \in (0,1)$. We only consider the last
  case, as the first two are simpler.  Hence, we can find $\nu'_1,
  \nu'_2 \in \DistrE$ such that $\delta(F_1) \alphaarrow \nu'_1$,
  $\delta(F_2) \alphaarrow \nu'_2$, $\mu'_1 \brbisim \nu'_1$,
  and~$\mu'_2 \brbisim \nu'_2$. From this it follows that
  $\delta(F_1+F_2) \alphaarrow \nu'$ and~$\mu' \brbisim \nu'$
  for~$\nu' = \nu'_1 \probc{r} \nu'_2$ using Lemma \ref{unrooted
    congruence}.

  Suppose $P_1 \rootedbrbisim Q_1$ and $P_2 \rootedbrbisim Q_2$ with
  $\calR_1$ and~$\calR_2$ rooted branching probabilistic bisimulation
  relations relating $\den{P_1}$ with~$\den{Q_1}$, and $\den{P_2}$
  with~$\den{Q_2}$, respectively, and fix some~$r\in(0,1)$. Then
  $\calR = \lc \langle {\mu}_1 \probc{r} {\mu}_2 , \nu_1 \probc{r}
  \nu_2 \rangle \mid {\mu}_1 \calR_1 \nu_1 ,\, {\mu}_2 \calR_2 \nu_2
  \rc$ is a rooted branching probabilistic bisimulation relation
  relating $\den{P_1 \probc{r} P_2}$ with~$\den{Q_1 \probc{r} Q_2}$.
  Symmetry is straightforward and decomposability can be shown along the lines
    of the proof of weak decomposability for Lemma~\ref{unrooted congruence}. 
    
  So we are left to prove the transfer property. 
  Suppose $({\mu}_1 \probc{r} {\mu}_2) \calR (\nu_1 \probc{r}
  \nu_2)$, thus $\mu_1 \calR_1 \mu_2$ and $\nu_1 \calR_2 \nu_2$. If
  $\mu_1 \probc{r} {\mu}_2 \alphaarrow {\mu'}$ then ${\mu}_1
  \alphaarrow \mu'_1$, $\mu_2 \alphaarrow \mu'_2$ and $\mu' = \mu'_1
  \probc{r} \mu'_2$ for suitable $\mu'_1, \mu'_2 \in \DistrE$. By
  assumption, $\nu'_1$ and~$\nu'_2$ exist such that $\nu_1 \alphaarrow
  \nu'_1$, $\nu_2 \alphaarrow \nu'_2$, $\mu'_1 \brbisim \nu'_1$,
  and~$\mu'_2 \brbisim \nu'_2$. From this we obtain $\nu_1 \probc{r}
  \nu_2 \alphaarrow \nu'$ and by Lemma~\ref{unrooted congruence}~$\mu'
  \brbisim \nu'$ for $\nu' = \nu'_1 \probc{r} \nu'_2$.
\end{proof}

\section{A few fundamental properties of branching bisimilarity}
\label{fundamental}

In this section we show two fundamental properties of branching probabilistic
bisimilarity that we need further on: the stuttering property, known
from~\cite{GW96:jacm} for non-deterministic processes, and
cancellativity of probabilistic choice with respect to ${\brbisim} \mkern1mu$.

\begin{lemma}[Stuttering Property]
  \label{stuttering}
  If $\mu \Arrow{} \bar\mu \Arrow{} \nu$ and $\mu\brbisim\nu$
  then $\mu\brbisim\bar\mu$.
\end{lemma}

\begin{proof}
  We show that the relation ${\brbisim} \mkern1mu \cup \{
  (\mu,\bar\mu), (\bar\mu,\mu) \}$ is a branching probabilistic
  bisimulation.

  First suppose $\mu \arrow\alpha \mu'$. Then there are
    $\nubar, \nu' \in \DistrE$ such that
    \begin{displaymath}
      \nu \Arrow{} \nubar, \ 
      \nubar \alphahidearrow \nu', \ 
      \mu \brbisim \nubar, \ \text{and} \ 
      \mu' \brbisim \nu'.
    \end{displaymath}
  Since $\bar\mu \Arrow{} \nu$, we have $\bar\mu \Arrow{} \nubar$,
  which had to be shown. 
  Now suppose $\bar\mu \arrow\alpha \mu'$.
  Then certainly $\mu \Arrow{}\bar\mu\arrow\alpha \mu'$.

  To show weak decomposability, suppose 
  $\mu = \bigoplusiinI \: p_i * \mu_i$.
  Then there are $\nubar, \nu_i \in \DistrE$, for~$i \in I$, such that
    \begin{displaymath}
      \nu \Arrow{} \nubar,\ 
      \mu \mopcalR \nubar,\  
      \nubar = \bigoplusiinI \: p_i * \nu_i,\ 
      \text{and} \ 
      \mu_i \mopcalR \nu_i 
      \ \text{for all~$i \in I$.}
    \end{displaymath}
  Again it suffices to point out that $\bar\mu \Arrow{} \nubar$.
  Conversely, suppose $\bar\mu = \bigoplusiinI \: p_i * \bar\mu_i$.
  Then $\mu  \Arrow{} \bar\mu = \bigoplusiinI \: p_i * \bar\mu_i$.
\end{proof}

\noindent
For $S \subseteq\calE$ and $\mu\in\DistrE$, define
$\mu(S)=\sum_{E\in\calE}\mu(E)$. Now two distributions $\mu$ and~$\nu$
are strong probabilistic bisimilar iff for each bisimulation
equivalence class $S\subseteq\calE$ one has $\mu(S)=\nu(S)$. The proof
is essentially the same as that of Lemma~\ref{decomp} below. However,
such a property does not hold for branching probabilistic
bisimilarity, due to the use of weak decomposability instead of
decomposability. But it does hold when restricting attention to a
class of processes on which weak decomposability reduces to
decomposability.

Call a distribution $\mu\in\DistrE$ \emph{$\brbisim$-stable} iff, for
all $\bar\mu\in\DistrE$, 
\begin{equation}\label{stable}
  \mu \Arrow{} \bar\mu \ \text{and} \ \mu \mathrel{\brbisim} \bar\mu
  \ \text{implies} \ \bar\mu = \mu\;,
\end{equation}
i.e., if it cannot perform internal activity without leaving its
branching bisimulation equivalence class. Note that if a distribution
$\mu \oplusr \nu$ with $r \in (0,1)$ is $\brbisim$-stable, then so are
$\mu$ and~$\nu$.

\begin{lemma}
  \label{decomp}
  If $\mu$ and $\nu$ are $\brbisim$-stable then $\mu\brbisim\nu$ iff
  $\mu(S)=\nu(S)$ for each $\brbisim$-equivalence class~$S$.
\end{lemma}

\begin{proof} Suppose $\mu\brbisim \nu$.
  Let $\mu = \bigoplus_{i\in I_1} \: p_i * E_i$. By weak
  decomposability, $\nu \Arrow{} \bar\nu = \bigoplusiinI \: p_i *
  \nu_i$ with $\nu \mathrel{\brbisim} \bar\nu$ and $\delta(E_i)
  \brbisim \nu_i$ for all $i\in I$.  By (\ref{stable}), as $\nu
  \brbisim \nubar$, we have $\bar\nu = \nu$.

  Let, for each $i\in I$, $\nu_i = \bigoplus_{j\in J_i} \: p_{ij} *
  F_{ij}$. Note, $\sum_{j \in J_i} \: p_{ij} = 1$.
  By weak decomposability, there are $\mu_{ij} \in \DistrE$,
  for~$j \in J_i$, such that $\delta(E_i) \Arrow{} \bar\mu_i {:=}
  \bigoplus_{i\in J_i} \: p_{ij} * \mu_{ij}$,
  $\nu_i \brbisim \bar\mu_i$ and $\mu_{ij} \brbisim
  \delta(F_{ij})$ for all~$j \in J_i$. By~(\ref{stable}) and
  $\brbisim$-stability of~$\delta(E_i)$, it follows that $\delta(E_i)
  = \bar\mu_i$ and hence $\mu_{ij} = \delta(E_i)$.
  Writing $E_{ij} {:=} E_i$, $q_{ij}:=p_i \cdot p_{ij}$ and $K = \lc
  ij \mid i \in I \land j\in J_i \rc$ we obtain
  \begin{displaymath}
    \mu = \bigopluskinK \: q_k * E_k \, ,\
    \nu = \bigopluskinK \: q_{k} * F_{k} \, , \ 
    \text{and} \ 
    \delta(E_k) \mopcalR \delta(F_k)
    \ \text{for all~$k \in K$.}
    \end{displaymath}
  Now, for any $\brbisim$-equivalence class $S \subseteq
    \calE$ it holds that $E_k\in S \Leftrightarrow F_k\in S$ for all
  $k\in K$. So, 
  \begin{displaymath}
    \mu(S) = \textstyle{\sum_{k \in K ,\, E_k \in S}} \: q_k = 
    \textstyle{\sum_{k \in K ,\, F_k \in S}} \: q_k = \nu(S) \, .
  \end{displaymath}
  The reverse direction of
  Lemma~\ref{decomp} is straightforward with Lemma~\ref{unrooted
    congruence}.
\end{proof}

\noindent
The next lemma holds because in this paper we consider finite
processes only.

\begin{lemma}
  \label{stabilizing}
  For each $\mu \in \DistrE$ there is a $\brbisim$-stable
  $\mu' \brbisim \mu$ with $\mu \Arrow{} \mu'$.
\end{lemma}

\begin{trivlist}
  \item[\hspace{\labelsep}\textbf{Proof Sketch:}] \hspace{-.6pt}Define the
    \emph{weight} of a distribution by $w(\mu) \mathbin=
    \sum_{E \in \calE} \: \mu(E)\cdot \mkern1mu c(E)$, i.e.,\ the weighted
    average of the complexities of the states in its support.  Now $E
    \arrow\alpha \mu$ implies $w(\mu) < w(\delta(E))$, and thus $\mu
    \arrow\alpha \mu'$ implies $w(\mu')< w(\mu)$, and $\mu \mathbin{\Arrow{}}
    \mu'$ implies $w(\mu') \mathbin\leqslant w(\mu)$. For $\mu \mathbin\in \DistrE$ let
    $T_\mu {:=} \lc \mu' \mid \mu' \brbisim \mu \land \mu \mathbin{\Arrow{}}
    \mu' \rc$ and define
    \begin{displaymath}
      sw(\mu) \: {:=} \: \textstyle{\inf_{\mu' \in T_\mu}} \: w(\mu') \, .
    \end{displaymath}
    In~\cite{DGHM09} it is shown that for any $\mu\in\DistrE$, the set
    $\lc \mu' \mid \mu \Arrow{} \mu' \rc$ is compact. This concept is
    defined by regarding a probability distribution over a finite set
    of $n$~states as a point in the $n$-dimensional Euclidean
    space~$\bbbr^n$. Similarly it can be shown that the set $T_\mu$~is
    compact. This is a consequence of the finitary nature of the space
    under consideration. The infimum over the compact set will be
    reached, and therefore there exists a derivation $\mu \Arrow{} \mu'$
    with $\mu' \brbisim \mu$ and $w(\mu') =
    sw(\mu)$. By construction $\mu'$ must be $\brbisim$-stable.  \qed
\end{trivlist}

\begin{lemma}[Cancellativity]
  \label{cancellative}
  Let $\mu, \mu', \nu, \nu' \in \DistrE$. If $\mu \oplusr \nu \brbisim
  \mu'\oplusr \nu'$ with $r \in (0,1]$ and $\nu\brbisim\nu'$, then
    $\mu \brbisim \mu'$.
\end{lemma}

\begin{proof}
  Let $\mu\oplusr\nu \brbisim \mu'\oplusr \nu'$ with
  $r \mathbin\in (0,1]$ and $\nu\brbisim\nu'$. By
  Lemma~\ref{stabilizing} there is a $\brbisim$-stable distribution
  $\xi \brbisim \mu \oplusr \nu$ with $\mu \oplusr \nu \Arrow{} \xi$.
  By weak decomposability, there are $\bar\xi \brbisim \xi$, $\bar\mu
  \brbisim \mu$ and $\bar\nu \brbisim \nu$ with $\xi \Arrow{} \bar\xi
  = \bar\mu \oplusr \bar\nu$. By the $\brbisim$-stability of $\xi$ we
  have $\bar\xi=\xi$.  Likewise there are distributions
  $\bar\mu'\brbisim\mu'$ and $\bar\nu'\brbisim \nu'$ such that
  $\bar\mu'\oplusr\bar\nu'$ is $\brbisim$-stable, $\mu'\oplusr \nu'
  \brbisim \bar\mu' \oplusr \bar\nu'$ and $\mu' \oplusr \nu' \Arrow{}
  \bar\mu'\oplusr \bar\nu'$.

  Since $\bar\mu \oplusr \bar\nu \brbisim \bar\mu'\oplusr \bar\nu'$,
  and using the $\brbisim$-stability of $\bar\mu \oplusr \bar\nu$ and
  $\bar\mu'\oplusr \bar\nu'$, it follows by Lemma~\ref{decomp} that
  $(\bar\mu \oplusr \bar\nu)(S) = (\bar\mu' \oplusr\bar\nu')(S)$
  for every $\brbisim$-equivalence class of states~$S$.
  Since $\bar\nu \brbisim \bar\nu'$, we likewise
  have $\bar\nu(S) = \bar\nu'(S)$. Using that $(\bar\mu\oplusr
  \bar\nu)(S) = r\cdot \bar\mu(S) + (1{-}r)\cdot\bar\nu(S)$,
  \begin{displaymath}
    \bar\mu(S) = 
    \frac{(\bar\mu\oplusr \bar\nu)(S) - (1{-}r)\cdot\bar\nu(S)}r
    = \frac{(\bar\mu'\oplusr\bar\nu')(S) - \bar\nu'(S)}r = 
    \bar\mu'(S).
  \end{displaymath}
  Thus, by Lemma~\ref{decomp}, $\bar\mu \brbisim \bar\mu'$,
  and consequently $\mu \brbisim \mu'$.
\end{proof}

\section{Completeness: the probabilistic case}

\label{sec-prob}

In this section we provide a sound and complete equational
characterization of rooted branching probabilistic bisimilarity.  The
completeness result is obtained along the same lines as the
corresponding result for branching {\bisimilarity} for the
non-deterministic processes in Section~\ref{sec-nondet}.  We extend
and adapt the non-deterministic theories $\AXd$ and~$\AXdb$ of
Section~\ref{sec-nondet}.

\begin{definition}[Axiomatization of $\bisim$ and~$\rootedbrbisim$]
  The theory~$\AXp$ is given by the axioms A1 to~A4, the axioms P1
  to~P3 and C listed in
  Table~\ref{table-axiomatization-of-probabilistic-branching-bisimulation}. 
  The theory~$\AXpb$ contains in addition the axioms~\hyperlink{BP}{BP}
  and~\hyperlink{G}{G}.
\end{definition}

\noindent
The axioms A1--A4 for non-deterministic processes are as before. Regarding
probabilistic processes, for the axioms P1 and~P2 dealing with
commutativity and associativity, we need to take care of the
probabilities involved. For~P2, it follows from the given restrictions
that also $(1{-}r)s = (1{-}r')s'$, i.e., the probability for~$Q$ to
execute is equal for the left-hand and right-hand side of the
equation. Axiom~P3 expresses that a probabilistic choice between equal
processes can be eliminated.
Axiom~C expresses that any two nondeterministic transitions can
be executed in a combined fashion: one with probability~$r$ and one
with the complementary probability~$1{-}r$.

The axioms P1 and P2 allow us to write each probabilistic
process~$P$ as 
\begin{displaymath}
  \partial(E_1)\probc{r_1}(\partial(E_2)\probc{r_2}(\partial(E_3)\probc{r_3}\cdots))   
\end{displaymath}
for non-deterministic processes $E_i$. In the sequel we denote such a
process by $\bigoplusiinI \: p_i * E_i$ with
$p_i=r_i\prod_{j=1}^{i-1}(1-r_j)$.
More specifically, if a probabilistic process~$P$ corresponds to a
distribution $\bigoplusiinI \: p_i * E_i$, then we have $\AXp \vdash P
= \bigoplusiinI \: p_i * E_i$, as can be shown by induction on the
structure of~$P$.

For axioms~\hyperlink{BP}{BP} and~\hyperlink{G}{G} of
Table~\ref{table-axiomatization-of-probabilistic-branching-bisimulation}
we introduce the notation~$E \sqsubseteq P$ for $E \in \calE$, $P \in
\calP$. We define
\begin{displaymath}
  E \sqsubseteq P
  \quad \text{iff} \quad
  \begin{array}[t]{l}
    \forall \alpha \in \calA , \, \mu \in \DistrE \colon
    E \arrow{\alpha} \mu \quad \Longrightarrow
    \smallskip \\
    \qquad\qquad\qquad
    \exists \mkern1mu \nu \in \DistrE \colon
    \den{P}  \arrow{(\alpha)} \nu \land \mu \brbisim \nu \, .
  \end{array}
\end{displaymath}
Thus, we require that every transition of the non-deterministic
process~$E$ can be directly matched by the probabilistic process~$P$.
Note, if $E \sqsubseteq P$ \vspace{1pt} and
$\delta(E) \arrow{\alpha} \mu$, then $\den{P} \alphahidearrow \nu$ for
some $\nu \in \DistrE$ such that $\mu \brbisim \nu \mkern1mu$: If
$\delta(E) \arrow\alpha \mu$, then
$\mu = \bigoplusiinI \: p_i * \mu_i$ and $E \arrow\alpha \mu_i$ for
suitable $p_i \geqslant 0$, $\mu_i \in \DistrE$. \vspace{1pt} Since
$E \sqsubseteq P$, we have for each~$i \in I$ that
$\den{P} \alphahidearrow \nu_i$ for some $\nu_i \in \DistrE$
satisfying $\mu_i \brbisim \nu_i$. Hence
$\den{P} \alphahidearrow \nu \, {:=} \, \bigoplusiinI \: p_i * \nu_i$ and
$\mu \brbisim \nu$ by Lemma~\ref{unrooted congruence}.

Axiom~\hyperlink{BP}{BP} is an adaptation of axiom~\hyperlink{B}{B} of
the theory~$\AXdb$ to the probabilistic setting of~$\AXpb$.  In the
setting of non-deterministic processes the implication $F
\arrow{\alpha} F'$ $\Longrightarrow$ $E \arrow{\alpha} E'$ $\land$ $F'
\brbisim E'$ for some $E'$ is captured by $E + F
\mathrel{\rootedbrbisim} E$. If we reformulate axiom~\hyperlink{B}{B}
as $E + F = E$ $\Longrightarrow$ $\alpha \pref ( F + \tau \pref E ) =
\alpha \pref E$, then it becomes more similar to
axiom~\hyperlink{BP}{BP} in
Table~\ref{table-axiomatization-of-probabilistic-branching-bisimulation}.

As to~\hyperlink{BP}{BP}, in the context of a preceding
action~$\alpha$ and a 
probabilistic process~$Q$, a non-deterministic alternative~$E$ that is
also offered by a probabilistic process after a $\tau$-prefix can be
dispensed with, together with the prefix~$\tau$. In a formulation
without the prefix~$\alpha$ and the probabilistic alternative~$Q$, but
with the specific condition $E \sqsubseteq P$, and retaining the
$\tau$-prefix on the right-hand side, the axiom~\hyperlink{BP}{BP} shows similarity
with axioms~T2 and~T3 in~\cite{FG19:jlamp} which, in turn, are
reminiscent of axioms T1 and~T2 of~\cite{DP07:tcs}; these axioms stem
from Milner's second $\tau$-law~\cite{Mil89:phi}.

\begin{table}
  \centering
  \def\arraystretch{1.1}
  \begin{tabular}{|@{\;}l@{\ \;}l|}
    \hline
    A1 & $E + F = F + E$ \rule{0pt}{12pt} \\
    A2 & $(E + F) + G = E +(F + G)$ \\
    A3 & $E + E = E$ \\
    A4 & $E + \bfzero = E$ 
    \rule[-5pt]{0pt}{12pt}
    \\ \hline
    P1 & $P \oplusr Q = Q \probc{1{-}r} P$ \rule{0pt}{12pt} \\
    P2 & $P \oplusr ( Q \probc{s} R ) = (P \probc{\bar{r}} Q)
    \probc{\bar{s}} R$ \\
    & \small 
    \qquad
    where $r = \bar{r} \mkern1mu \bar{s}$ and $(1{-}r)(1{-}s) =
    1{-}\bar{s}$ \quad \\ 
    P3  & $P \oplusr P = P$
    \rule[-5pt]{0pt}{12pt}
    \\ \hline
    \rule{0pt}{12pt}%
    \hypertarget{C}{C}
    & $\alpha \pref P + \alpha \pref Q = \alpha \pref P + \alpha
    \pref ( P \oplusr Q ) + \alpha \pref Q$   
    \rule[-5pt]{0pt}{12pt}
   \\ \hline
    \hypertarget{BP}{BP} & if $E \sqsubseteq P$ then    
    \rule{0pt}{12pt}\\
    &
    \qquad
    $\alpha \pref \bigl (
    \partial ( \, E + \tau \pref P \, ) 
    \oplusr Q \bigr ) = 
    \alpha \pref \bigl ( P \oplusr Q \bigr )$
    \rule[-5pt]{0pt}{12pt}
    \\
    \hypertarget{G}{G} &  if $E \sqsubseteq \partial(F)$ then    
    \rule{0pt}{12pt}\\
    &
    \qquad
      $\alpha\pref (\partial(E+F) \oplusr Q) =
      \alpha\pref (\partial(F) \oplusr Q)$
    \rule[-5pt]{0pt}{12pt} \\
    \hline
  \end{tabular}
  \def\arraystretch{1.0}

  \medskip

  \caption{Axioms for strong and rooted branching probabilistic {\bisimilarity}}
  \label{table-axiomatization-of-probabilistic-branching-bisimulation}
\end{table}

Let us illustrate the working of axiom~\hyperlink{BP}{BP}. Consider the
non-deterministic process $E = \bstop$ and the probabilistic process
$P = \partial( \astop + \bstop ) \probc{1/2} \partial( \bstop)$. Then we
have $E \sqsubseteq P$, i.e.\ 
\begin{displaymath}
  \bstop \sqsubseteq 
  \partial( \astop + \bstop ) \probc{1/2} \partial( \bstop ).
\end{displaymath}
Therefore, we have by application of axiom~\hyperlink{BP}{BP} the
provable equality
\begin{displaymath}
  \begin{array}{l}
    \AXpb \vdash
    \alpha \pref \bigl( \partial \dg( {\bstop} + \tau \pref \mg( 
    \partial\dor( \astop + \bstop \dor) \probc{1/2}
    \rule{0pt}{9.5pt}\partial\dor( 
    \bstop \dor) \mg) 
     \dg) \oplusr Q \bigr)
    = {} \smallskip \\ \hspace{94pt}
    \alpha \pref \bigl( \dg(\rule{0pt}{9.5pt}
    \partial\dor( \astop + \bstop \dor) \probc{1/2} \partial\dor( \bstop \dor)
    \dg) \oplusr Q \bigr) \; .
  \end{array}
\end{displaymath}
Another example is
$a \pref (P_1\oplusr P_2) \sqsubseteq \partial\dg(b\pref R + a\pref
P_1\dg)\oplusr\partial\dg(c\pref S + a\pref P_2\dg)$, 
so
\begin{displaymath}
  \begin{array}{l}
    \AXpb \vdash
    \alpha \pref \bigl(\partial\mg(
    a\pref(P_1\oplusr P_2) + \tau \pref \dor( \partial\dg( b \pref R
    + 
    a \pref P_1\dg)\oplusr\partial\dg(c\pref S + a\pref P_2\dg) \dor) 
     \mg) \oplusr Q\bigr)
    = {} \smallskip \\ \hspace{121pt}
    \alpha \pref \bigl( \dor( \partial\dg(b\pref R + a\pref P_1 \dg) \oplusr
    \partial\dg(c\pref S + a\pref P_2\dg) \dor)
     \oplusr Q
    \bigr) \;.
  \end{array}
\end{displaymath}
An example illustrating the use of $(\alpha)$, rather than
$\alpha$, as label of the matching transition of~$\den{P}$ in the
definition of $\sqsubseteq$ is
\begin{displaymath}
  \tau \pref \mg(\partial(b \pref P + \tau \pref Q)\oplusr Q\mg)
  \sqsubseteq  \partial(b \pref P + \tau \pref Q)
\end{displaymath}
from which we obtain
\begin{displaymath}
  \begin{array}{l}
    \AXpb \vdash
    \alpha \pref \left(\partial\rule{0pt}{9.5pt}\big(
    \tau \pref \mg(\partial(b \pref P + \tau \pref Q)\oplusr Q\mg)
    + \tau \pref \mg( \partial(b \pref P + \tau \pref Q) \mg) 
    \big)\oplusr R \right)
    = {} \smallskip \\ \hspace{173pt}
    \alpha \pref\big( \mg( \partial(b \pref P + \tau \pref Q) \mg)\oplusr R\big)\;.
  \end{array}
\end{displaymath}
The axiom \hyperlink{G}{G} roughly is a variant of \hyperlink{BP}{BP}
without the $\tau$ prefixing the process $P$. A typical example,
matching the one above, is
\begin{displaymath}
  \begin{array}{l}
    \AXpb \vdash
    \alpha \pref \left(\partial\rule{0pt}{9.5pt}\big(
    \tau \pref \mg(\partial(b \pref P + \tau \pref Q)\oplusr Q \mg)
    + \mg( b \pref P + \tau \pref Q \mg) 
    \big)\oplusr R \right) 
    = {} \smallskip \\ \hspace{153pt}
    \alpha \pref \big( \mg( \partial(b \pref P + \tau \pref Q)
    \mg) \oplusr R \big) \; . 
  \end{array}
\end{displaymath}
The occurrences of the prefix $\alpha \pref {\_}$
in~\hyperlink{BP}{BP} and~\hyperlink{G}{G} are related to the
root condition for non-deterministic processes,
cf.\ axiom~\hyperlink{B}{B} in Section~\ref{sec-nondet}.

\begin{lemma}
  \label{simpleBP}
  The following simplifications of the axiom \hyperlink{BP}{BP} are derivable:
  \begin{displaymath}
    \def\arraystretch{1.2}
    \begin{array}{rl}
      (i) &
      \AXpb\vdash \alpha \pref 
      \partial ( \, E + \tau \pref P \, ) = 
      \alpha \pref  P \quad \text{if $E \sqsubseteq P$,} \\
      (ii) &
      \AXpb\vdash \alpha \pref 
      (\partial ( \, \tau \pref P) \oplusr R \, ) = 
      \alpha \pref   (P \oplusr R) \textrm{ and} \\
      (iii) &
      \AXpb\vdash \alpha \pref 
      \partial ( \, \tau \pref P)  = 
      \alpha \pref   P \, .
    \end{array}
    \def\arraystretch{1.0}
  \end{displaymath}
\end{lemma}

\begin{proof}
  The proof of the first equality requires an application of~P3:
  \[
  \def\arraystretch{1.2}
  \begin{array}{ll}
    \AXpb \vdash &
    \alpha \pref 
    \partial ( E + \tau \pref P ) ~\stackrel{\textrm{P3}}{=} \\
    & \alpha \pref  \bigl( 
    \partial ( E + \tau \pref P ) \probc{\frac12}
    \partial ( E + \tau \pref P ) \bigr)
    ~\stackrel{\textrm{\hyperlink{BP}{BP}}}{=} \\
    & \alpha \pref 
    \bigl( 
    P \probc{\frac12} \partial ( E + \tau \pref P )
    \bigr)
    ~\stackrel{\textrm{P1}}{=} \\
    & \alpha \pref 
    \bigl(
    \partial ( E + \tau \pref P ) \probc{\frac12} P
    \bigr )
    ~\stackrel{\textrm{\hyperlink{BP}{BP}}}{=} \\
    & \alpha \pref 
    \bigl( P \probc{\frac12} P \bigr )
    ~\stackrel{\textrm{P3}}{=}  
    \alpha \pref P \, .
  \end{array}
  \def\arraystretch{1.0}\]
  The proof of the second equality uses A4:
  \[
  \def\arraystretch{1.2}
  \begin{array}{ll}
    \AXpb \vdash &
    \alpha \pref \bigl(
    \partial ( \tau \pref P ) \oplusr R  
    \bigr) ~\stackrel{\textrm{A4}}{=} \\
    & \alpha \pref \bigl(
    \partial ( {\bfzero} + \tau \pref P ) \oplusr R  
    \bigr) ~\stackrel{\textrm{\hyperlink{BP}{BP}}}{=} \\
    &\alpha \pref \bigl( 
    P \oplusr R  
    \bigr) 
\end{array}\]
The last equation can be proven in a similar way as the first
property, using the second one.
\end{proof}

\noindent
Similar simplifications of axiom \hyperlink{G}{G} can be found.

Using the properties in the lemma above the process identities
mentioned in the introduction can easily be proven. Returning to the
processes $E_1$ and~$E_6$ related to Figure~\ref{PA-of-Eisentraut}, we
have
\begin{displaymath}
  \def\arraystretch{1.2}
  \begin{array}{rcl}
    \AXpb \vdash 
    E_1 & = & \tau \pref \bigl(
  \partial\dg( 
    \tau \pref \partial\dor( \tau \pref P + c \pref Q + \tau \pref R \dor) +
    c \pref Q + \tau \pref R \dg) 
    \\ & & \phantom{\tau \pref \bigl( {}} 
  {} \probc{1/2}
  \partial\dg( \tau \pref \mg(
    \partial\dor( \tau \pref P + c \pref Q + \tau \pref R \dor)
    \probc{1/2}
    \partial( \bfzero ) \mg) \dg)
  \bigr) \\
  & \stackrel{\text{\hyperlink{BP}{BP}}}{=} & 
  \tau \pref \bigl(
  \partial\dor( 
    \tau \pref P + c \pref Q + \tau \pref R \dor) 
  \\ & & \phantom{\tau \pref \bigl( {}} 
  {} \probc{1/2}
  \partial\dg( \tau \pref \mg(
    \partial\dor( \tau \pref P + c \pref Q + \tau \pref R \dor)
    \probc{1/2}
    \partial( \bfzero ) \mg) \dg)
  \bigr) \\
  & \stackrel{\text{\ref{simpleBP} $(ii)$, P1}}{=} & 
  \tau \pref \bigl(
  \partial\dor( \tau \pref P + c \pref Q + \tau \pref R \dor) 
  \\ & & \phantom{\tau \pref \bigl( {}} 
  {} \probc{1/2}
  \mg(
    \partial\dor( \tau \pref P + c \pref Q + \tau \pref R \dor)
    \probc{1/2}
    \partial( \bfzero )  \mg)
  \bigr) \\
  & \stackrel{\text{P2}}{=} & 
  \tau \pref \bigl(
   \mg(\partial\dor(\tau\pref P + c \pref Q + \tau \pref R \dor) 
  \probc{2/3}\\
&&  \phantom{\tau \pref \bigl( ({}}  \partial\dor( \tau \pref P + c \pref Q + \tau \pref R \dor)\mg)
  \probc{3/4}
  \partial( \bfzero ) 
  \bigr)
  \\
  & \stackrel{\text{P3}}{=} & 
  \tau \pref ( \partial\dor( \tau \pref P + c \pref Q + \tau \pref R \dor) 
  \probc{3/4} \partial( \bfzero ) )
  = E_6 \, .
  \end{array}
  \def\arraystretch{1.0}
\end{displaymath}

\noindent
Soundness of the theory~$\AXp$ for strong probabilistic {\bisimilarity}
and of the theory~$\AXpb$ for rooted branching probabilistic
{\bisimilarity} is straightforward.

\begin{lemma} [Soundness]
  \label{lemma-soundness}
  For all $P,Q \in \calP$, if $\AXp \vdash P = Q$ then $P \bisim Q$,
  and if $\AXpb \vdash P = Q$ then $P \rootedbrbisim Q$.
\end{lemma}

\begin{proof}
  As usual, in view of $\bisim$ and~$\rootedbrbisim$ being
  congruences, one only needs to prove the left-hand and right-hand
  sides of the axioms to be strongly or rooted branching
  probabilistically bisimilar. We only treat the cases
  of the axioms~\hyperlink{BP}{BP} and~\hyperlink{G}{G}
  with respect to rooted branching probabilistic bisimilarity.

  For \hyperlink{BP}{BP}, by Definition~\ref{rooted-bpb} and
  Lemma~\ref{unrooted congruence}, it suffices to show that
  $P \brbisim \partial(E + \tau \pref P)$ if~$E \sqsubseteq
  P$. Suppose $\delta( E + \tau \pref P ) \alphaarrow \mu$. We
  distinguish two cases: (i)~%
  $\delta(E) \alphaarrow \mu$; (ii)~$\alpha = \tau$,
  $\delta(E) \arrow{\tau} \mu'$ and $\mu = \den{P} \probc{r} \mu'$ for
  some~$r \in (0,1]$. For~(i), by definition of~$E \sqsubseteq P$, we
  have $\den{P} \arrow{(\alpha)} \nu$ and $\mu \brbisim \nu$ for
  suitable~$\nu \in \DistrE$. For~(ii), again
  by~$E \mathbin\sqsubseteq P$, we have $\den{P} \arrow{(\tau)} \nu'$
  and $\mu' \brbisim \nu'$. Thus
  $\den{P} \mathbin= \den{P} \!\probc{r}\! \den{P} \arrow{(\tau)} \nu$
  \linebreak[3] and $\mu \brbisim \nu$ for
  $\nu = \den{P} \probc{r} \nu'$, as was to be shown.
  Conversely, $\den{P} \alphaarrow \mu$ trivially implies that
  $\delta(E + \tau \pref P) \Arrow{} \den{P} \alphahidearrow \mu$.
  The requirement on weak decomposability also holds trivially.

  For \hyperlink{G}{G}, by Definition~\ref{rooted-bpb} and
  Lemma~\ref{unrooted congruence}, it suffices to show that $E+F
  \brbisim F$ if $E \sqsubseteq \partial(F)$. Put $\calR = \singleton{
    \langle E+F , F \mkern1mu \rangle }^\dagger \cup {\brbisim}
  \mkern1mu$. We verify that $\calR$ is a branching probabilistic
  bisimulation. Naturally, $\delta(E) \alphaarrow \mu$ implies
  $\delta(E+F) \alphaarrow \mu$, and also weak decomposability is
  easy. Finally, suppose $\delta(E+F) \alphaarrow \mu$. Since $E
  \sqsubseteq \partial(F)$ \vspace{2pt} now we have $\delta(E)
  \alphahidearrow \nu$ for some~$\nu$ with $\mu \brbisim \nu$.
\end{proof}

\noindent
As for the process language with non-deterministic processes only, we
aim at a completeness proof that is built on completeness of strong
{\bisimilarity} and the notion of a concrete process.
Equational characterization of strong probabilistic {\bisimilarity}
has been addressed by various authors. The theory~$\AXp$ provides a
sound and complete theory. For a proof, see e.g.~\cite{Hen12:facj}.

\begin{lemma}
  \label{lemma-completeness-for-strong-probabilistic-bisimulation}
  The theory $\AXp$ is sound and complete for strong {\bisimilarity}.
\end{lemma}

\noindent
The next lemma provides a more state-based characterization of strong
probabilistic bisimilarity.

\begin{lemma}\label{state based}
  Let $\calR \subseteq \DistrE \times \DistrE$ be a decomposable
  relation such that 
  \begin{equation}
    \label{cong}
    \mu_1 \mopcalR \nu_1 \ \text{and} \ \mu_2 \mopcalR \nu_2
    \quad  \text{implies} \quad
    (\mu_1 \oplusr \mu_2) \mopcalR (\nu_1 \oplusr \nu_2)
  \end{equation}
  and for each pair $E,F \in \calE$ 
  \begin{equation}
    \label{diamond}
    \delta(E) \mopcalR \delta(F)
    \ \text{and} \
    E \arrow{\alpha} \mu'
    \quad \text{implies} \quad
    \delta(F) \alphaarrow \nu'
    \ \text{and} \ 
    \mu' \mopcalR \nu'
  \end{equation}
  for a suitable $\nu' \in \DistrE$.
  Then $\mu \mopcalR \nu$ implies $\mu\bisim \nu$.
\end{lemma}

\begin{proof}
  We show that $\calR$ is a strong probabilistic bisimulation relation.
  So, let $\mu, \nu \in \DistrE$ such that $\mu \mopcalR \nu$ and
  $\mu \arrow{\alpha} \mu'$.  By
  Definition~\ref{def-pr-transition-relation}(b) we have
  $\mu = \bigoplusiinI \: p_i * E_i$,
  $\mu' = \bigoplusiinI \: p_i * \mu'_i $, and
  $E_i \alphaarrow \mu'_i$ for all $i \in I$.  Since $\calR$ is
  decomposable, there are $\nu_i \in \DistrE$, for~$i \in I$, such
  that
    \begin{displaymath}
      \nu  = \bigoplusiinI \: p_i * \nu_i
      \quad \text{and} \quad
      \delta(E_i) \mopcalR \nu_i 
      \ \text{for all~$i \in I$.}
    \end{displaymath}
    Let, for each $i\in I$,
    $\nu_i = \bigoplus_{j\in J_i} \: p_{ij} * F_{ij}$.  Since $\calR$
    is decomposable, there are $\mu_{ij} \in \DistrE$,
    for~$j \in J_i$, such that
    \begin{displaymath}
      \delta(E_i)  = \textstyle\bigoplus_{i \in J_i} \: p_{ij} * \mu_{ij}
      \quad \text{and} \quad 
      \mu_{ij} \mopcalR \delta(F_{ij})
      \ \text{for all~$j \in J_i$.}
    \end{displaymath}
    Here $\mu_{ij} = \delta(E_i)$.  Writing $E_{ij} \, {:=} \, E_i$,
    $q_{ij} \, {:=} \, p_i\cdot p_{ij}$ and
    $K = \lc (i,j) \mid i \in I \land j \in J_i \rc$ we obtain
    \begin{displaymath}
      \mu= \bigopluskinK \: q_k * E_k \ ,\
      \nu = \bigopluskinK \: q_{k} * F_{k}
      \quad \text{and} \quad
      \delta(E_k) \mopcalR \delta(F_k)
      \ \text{for all~$k \in K$.}
    \end{displaymath}
    Let $\mu'_{ij} {:=} \mu'_i$ for all $i \in I$ and $j \in
    J_i$. Then $\mu' = \bigopluskinK \: q_{k} * \mu'_k$. Using
    that $E_k \arrow{\alpha} \mu'_k$ for all $k \in K$, there must be
    distributions~$\nu'_k$ for $k \in K$ such that
    \begin{displaymath}
      \delta(F_k) \alphaarrow \nu'_k
      \quad \text{and} \quad
      \mu'_k \mopcalR \nu'_k.
    \end{displaymath}
    By Definition~\ref{def-pr-transition-relation}(b) this implies
    $\nu \alphaarrow \nu'$, for
    $\nu' \, {:=} \, \bigopluskinK \: q_{k} * \nu'_k$.  Moreover,
    (\ref{cong})~yields $\mu' \mopcalR \nu'$.
\end{proof}

\noindent
The following technical lemma expresses that two rooted branching
probabilistically bisimilar processes can be represented in a similar
way.

\begin{lemma}
  \label{lemma-form-of-probabilistic-branching-processes}
  For all $Q, R \in \calP$, if $Q \rootedbrbisim R$ then there are an
  index set~$I$ as well as for all $i\in I$ suitable $p_i > 0$ and
  $F_i, G_i \in \calE$ such that $F_i \rootedbrbisim G_i$, $\AXp \vdash Q =
  \bigoplusiinI \: p_i * F_i$ and $\AXp \vdash R = \bigoplusiinI \:
  p_i * G_i$.
\end{lemma}

\begin{proof}
  Suppose $\AXp \vdash Q = \bigoplusjinJ \: q_j * F'_j$. Since
  $Q \rootedbrbisim R$ and $\rootedbrbisim$ is decomposable, the
  process~$R$ can be written as $\bigoplusjinJ \: q_j * R_j$ with
  $R_j \in \calP$ for $j \in J$, such that
  $\partial(F_j) \rootedbrbisim R_j$. Therefore, each
  distribution~$R_j$ can be written as
  $\bigoplus_{k {\in} K_j} \: r_{jk} * G_{jk}$ where
  $\partial(F'_j) \rootedbrbisim \partial(G_{jk})$ for all $j \in J$
  and $k \in K_j$. We now define $F_{jk} = F'_j$ for $j \in J$ and
  $k \in K_j$. Then, using the axioms P1, P2 and P3 we can derive
  \begin{displaymath}
    \begin{array}{l}
    \AXp \vdash Q = \bigoplusjinJ \, \bigoplus_{k {\in} K_j} \: q_j r_{jk}  * F_{jk}
    \smallskip \\ 
    \AXp \vdash R = \bigoplusjinJ \, \bigoplus_{k {\in} K_j} \: q_j r_{jk}  *
    G_{jk}
    \end{array}
  \end{displaymath}
   with $F_{jk} \brbisim G_{jk}$ for $j \in J$ and $k \in K_j$. This
   proves the lemma. 
\end{proof}

\noindent
Similar to the non-deterministic case, a transition $E \arrow{\tau}
\mu$ is called \textit{inert} iff $\delta(E) \brbisim \mu$. Typical
cases of inert transitions include
\begin{displaymath}
  \begin{array}{c}
    \tau \pref P \  \arrow{\tau} \  \den{P} \;, \\[2pt]
    E + \tau \pref \partial(E) \  \arrow{\tau} \  \delta(E) \;.
  \end{array}
\end{displaymath}
Furthermore, a transition $E \arrow{\tau} \mu_1\oplusr \mu_2$ with
$r \in (0,1]$ and $\delta(E) \brbisim \mu_1$ is called
\emph{partially inert}. A typical case is
\begin{displaymath}
  \begin{array}{c}
    \tau \pref (\partial(b \pref P + \tau \pref Q)\oplusr Q)
    + b \pref P + \tau \pref Q
    \  \arrow{\tau} \
    \delta(b \pref P + \tau \pref Q)\oplusr \den{Q} \;.
  \end{array}
\end{displaymath}
Here $\delta \big( \tau \pref (\partial(b
\pref P + \tau \pref Q)\oplusr Q) + b \pref P + \tau \pref Q \big)
\brbisim \delta(b \pref P + 
\tau \pref Q)$ because\linebreak $\delta(b \pref P + \tau \pref Q)
\tauhidearrow \delta(b \pref P + \tau \pref Q)\oplusr \den{Q}$.

In Section~\ref{sec-nondet} a process is called concrete if it does
not exhibit an inert transition. In the setting with probabilistic
choice we need to be more careful. For example, we also want to
exclude processes of the form
\begin{displaymath}
  \partial( \tau \pref P ) \probc{1/2} \partial( a \pref Q ) )
  \quad \text{and} \quad
  \partial( a \pref P ) \probc{1/2} \partial( b \pref \bigl(
  \partial( \tau \pref Q ) \probc{1/3} Q
  \bigr))
\end{displaymath}
from being concrete, although they cannot
perform a transition by themselves at all. Therefore, we define
the \emph{derivatives} $\der(P)\subseteq\calE$ of a probabilistic
process $P\in\calP$ by
\begin{displaymath}
  \begin{array}{r@{~:=~}l}
    \der(P \oplusr Q) & \der(P) \cup \der(Q) 
    \smallskip \\
    \der(\partial(\sum_{i\in I}\alpha_i\cdot P_i)) &
    \{ \, \sum_{i\in I} \: \alpha_i\cdot P_i \, \} \cup 
    \bigcup_{i\in I} \: \der(P_i)
  \end{array}
\end{displaymath}
and define a process~$\Pbar \in \calP$ to be \emph{concrete} iff none
of its derivatives can perform a partially inert transition, i.e., if
there is no transition $E \arrow{\tau} \mu_1\oplusr \mu_2$ with $E \in
\der(\Pbar)$, $r \in (0,1]$ and $\delta(E) \brbisim \mu_1$. A
non-deterministic process~$\Ebar$ is called concrete if the
probabilistic process~$\partial(\Ebar)$ is. Moreover, we define two
sets of concrete processes:
\begin{displaymath}
  \calEconc = \lc \Ebar \in \calE \mid \text{$\Ebar$ is concrete} \rc
  \quad \text{and} \quad \calPconc = \lc \Pbar \in \calP \mid
  \text{$\Pbar$ is concrete} \rc.
\end{displaymath}
Furthermore, we call a process $E \in \DistrE$ \emph{rigid} iff there
is no inert transition $E \arrow\tau \mu$, and write $\calE_r = \lc
\Ebar \in \calE \mid \text{$\Ebar$ is rigid} \rc$. Naturally,
$\calEconc \subseteq \calE_r$.

We use concrete and rigid processes to build the proof of the
completeness result for rooted branching probabilistic {\bisimilarity}
on top of the completeness proof of strong probabilistic
{\bisimilarity}. The following lemma lists all properties of concrete
and rigid processes we need in our completeness proof.

\begin{lemma}
  \label{cc}
  \begin{enumerate}[(a)]
  \item If $E = \sumiinI \: \alpha_i \cdot P_i$ with $P_i \in
    \calPconc$ and, for all $i \in I$, $\alpha_i \neq \tau$ or $\den{P_i}$
    cannot be written as $\mu_1 \oplusr \mu_2$ with $r
    \mathbin \in (0,1]$ and $\delta(E) \brbisim \mu_1$, then $E \in
      \calEconc$.
  \item If $P_1,P_2\in\calPconc$ then $P_1 \oplusr P_2\in\calPconc$.
  \item If $\mu = \bigoplusiinI \: p_i * \mu_i \in \Distr(\calEconc)$
    with each $p_i>0$, then each $\mu_i\in\DistrEconc$.
  \item If $\mu\in\Distr(\calEconc)$ and $\mu \alphahidearrow \mu'$
    then $\mu'\in\DistrEconc$. 
  \item If $\mu \in \Distr(\calE_r)$ and $\mu \,\Arrow{}\, \mu'$ with
    $\mu \brbisim \mu'$ then $\mu = \mu'$.
  \item If $E \mathbin\in \calEconc$, $F \mathbin\in \calE$, and
    $\mu,\nu \mathbin\in \DistrE$ are such that \vspace{2pt} $E
    \brbisim F$, 
    $E \arrow{\alpha} \mu$, $\delta(F) \arrow{(\alpha)} \nu$ and
    $\mu \brbisim \nu$, then $\delta(F) \arrow{\alpha} \nu$.
  \item If $\mu \brbisim \bigoplusiinI \: p_i * \nu_i$ for $\mu \in
    \Distr(\calE_r)$ then $\mu = \bigoplusiinI \: p_i*\mu_i$ for
    certain $\mu_i \brbisim \nu_i$.
  \end{enumerate}
\end{lemma}

\begin{proof}
  Properties (a), (b), (c) and~(d) follow immediately from the
  definitions, in the case of~(d) also using
  Definition~\ref{def-pr-transition-relation}(b).

  For (e), let $\mu\mathbin\in\Distr(\calE_r)$ and $\mu \Arrow{} \mu'$
  with $\mu\brbisim\mu'$. Towards a contradiction, suppose
  $\mu\neq\mu'$.  Then there must be a distribution $\bar\mu\neq\mu$
  such that $\mu \tauhidearrow \bar\mu$ and $\bar\mu \Arrow{} \mu'$.
  We may even choose $\bar\mu$ such that the transition
  $\mu \tauhidearrow \bar\mu$ acts on only one (rigid) state in
  the support of $\mu$, i.e.\ there are $E \in \calE$,
  {$r \in (0,1]$} and $\rho,\nu \in \DistrE$ such
  that $\mu = \delta(E) \oplusr \rho$, $E \arrow\tau \nu$ and
  $\bar\mu = \nu \oplusr \rho$. By Lemma~\ref{stuttering}
  $\delta(E) \oplusr \rho = \mu \brbisim \bar\mu = \nu \oplusr \rho$.
  Hence by Lemma~\ref{cancellative} $\delta(E) \brbisim \nu$. So the
  transition $E \arrow\tau \nu$ is inert, contradicting $E \in \calE_r$.

  To establish (f), suppose $E \in \calEconc$, $F \in \calE$, and
  $\mu,\nu \in \DistrE$ are such that $E \brbisim F$,
  $E \arrow{\alpha} \mu$, $\delta(F) \arrow{(\alpha)} \nu$ and
  $\mu \brbisim \nu$. Assume $\alpha = \tau$, for otherwise the
  statement is trivial. Then $\delta(F) \tauarrow \nu_1$ and
  $\nu = \nu_1 \oplusr \delta(F)$ for some $\nu_1 \in \DistrE$ and
  $r \in [0,1]$. Since $\brbisim$ is weakly
  decomposable, there are $\bar\mu, \mu_1, \mu_2 \in \DistrE$ such
  that $\mu \Arrow{} \bar\mu$, $\mu \brbisim \bar\mu$,
  $\bar\mu= \mu_1 \oplusr \mu_2$, $\mu_1 \brbisim \nu_1$ and
  $\mu_2 \brbisim \delta(F)$. Since $\mu$ is concrete, using case~(d)
  of the lemma, $\bar\mu = \mu$ by case~(e). Thus
  $E \arrow{\tau} \mu_1\oplusr \mu_2$ with
  $\delta(E) \brbisim \delta(F) \brbisim \mu_2$.  Since $E$~is
  concrete, this transition cannot be partially inert. Thus, we
    must have~$r = 1$. It follows that
  $\delta(F) \arrow{\alpha} \nu$.

  Regarding~(g), if $\mu \brbisim \bigoplusiinI \: p_i * \nu_i$ for
  $\mu \mathbin\in \Distr(\calE_r)$, then $\mu \Arrow{} \bar\mu :=
  \bigoplusiinI \: p_i * \mu_i$ with $\mu_i \brbisim \nu_i$
  by weak decomposability of~$\brbisim \mkern1mu$.
  By~(e) we have $\bar\mu = \mu$.
\end{proof}

\begin{lemma}
  \label{lemma-equal-for-concrete-processes}
  For all $\Pbar,\Qbar \in \calPconc$, if $\Pbar \brbisim \Qbar$ then
  $\Pbar \bisim \Qbar$ and $\AXp \vdash \Pbar = \Qbar$.
\end{lemma}

\begin{proof}
  Let $\calR \, {:=} \, {\brbisim} \cap (\Distr(\calEconc) \times
  \Distr(\calEconc))$.  Then, by Lemma~\ref{cc}(c)--(d), $\calR$ is a
  branching probabilistic bisimulation relation relating $\Pbar$
  and~$\Qbar$.  We show that $\calR$ moreover satisfies the conditions
  of Lemma~\ref{state based}.  Condition~(\ref{cong}) is a direct
  consequence of Lemmas~\ref{unrooted congruence} and~\ref{cc}(b).
  That $\calR$ is decomposable follows since it is weakly
  decomposable, in combination with Lemma~\ref{cc}(e). Now, in
    order to verify condition~(\ref{diamond}), suppose
  $\delta(E) \mopcalR \delta(F)$ and $E\alphaarrow \mu$.  Then
  $\delta(F) \Arrow{} \nubar \alphahidearrow \nu$ for some
  $\nubar,\nu\in\calP$ with $\delta(E) \mopcalR \nubar$ and
  $\mu \mopcalR \nu$. By Lemma~\ref{cc}(e) we
  have~$\nubar = \delta(F)$.
  Thus $\delta(F) \arrow{(\alpha)} \nu$. Hence, using Lemma~\ref{cc}(f)
  it follows that $\delta(F) \arrow{\alpha} \nu$.
  With $\calR$ satisfying conditions (\ref{cong}) and~(\ref{diamond}),
  Lemma~\ref{state based} yields $\Pbar \bisim \Qbar$. By
  Lemma~\ref{lemma-completeness-for-strong-probabilistic-bisimulation}
  we obtain $\AXp \vdash \Pbar = \Qbar$.
\end{proof}

\noindent
Before we are in a position to prove our main result we need one more
technical lemma. Here we write $\AXpb \vdash P_1 \approx P_2$ as
  a shorthand for
  \begin{displaymath}
    \forall \alpha \in \calA \,
    \forall Q \in\calP \,
    \forall r \in (0,1) \colon
    \AXpb \vdash \alpha \pref (P_1 \oplusr Q) =
    \alpha \pref (P_2 \oplusr Q) \, .
  \end{displaymath}
For example, using axiom \hyperlink{BP}{BP}, if $E\sqsubseteq P$ then
$\AXpb \vdash \partial (E + \tau \pref P ) \approx P$.  Likewise,
using \hyperlink{G}{G}, if $E \sqsubseteq\partial(F)$ then
$\partial(E+F)\approx\partial(F)$.  As in the proof of
Lemma~\ref{simpleBP}(i), from $\AXpb \vdash P_1 \approx P_2$ it also
follows that $\AXpb \vdash \alpha \pref P_1 \approx \alpha \pref P_2$
for all $\alpha \in \calA$.
In the proof of the lemma we rely on the axioms
\hyperlink{BP}{BP} and~\hyperlink{G}{G}.

\begin{lemma}
  \label{lemma-prob-technical}
  \mbox{}
  \begin{itemize}
  \item [(a)] For each non-deterministic process $E \in \calE$ there
    is a concrete probabilistic process $\Pbar \in \calPconc$ such
    that $\AXpb \vdash \partial(E) \approx \Pbar$.

  \item [(b)] For each probabilistic process $P \in \calP$ there is a
    concrete probabilistic process $\Pbar \in \calPconc$ such that
    $\AXpb \vdash { \alpha \pref P =
      \alpha \pref \Pbar }$ for all $\alpha \in \calA$.

  \item [(c)] For all probabilistic processes $Q,R \in \calP$, if $Q
    \brbisim R$ then $\AXpb \vdash { \alpha \pref Q = \alpha \pref R
    }$ for all~$\alpha \in \calA$.
  \end{itemize}
\end{lemma}

\begin{proof}
  By simultaneous induction on $c(E)$, $c(P)$, and
  $\max \lbrace c(Q), c(R) \rbrace$. The base case, which applies to
  case (a) only, is clear, since the process $\bfzero+\cdots+\bfzero$
  is concrete. \pagebreak[3]

  Case (a) for $c(E) > 0$. The process~$E$ can be written as $\sumiinI
  \: \alpha_i \pref P_i $ for some index set~$I$ and suitable~$\alpha_i
  \in \calA$ and $P_i \in \calP$. Pick by the induction
  hypothesis~(b), for each $i \in I$, a concrete probabilistic
  process~$\Pbar_i \mathbin\in \calPconc$ with
  $\AXpb \vdash { \alpha_i \pref \Pbar_i = \alpha_i \pref P_i }$.
  Now $\AXp \vdash E = \bar{E}$ for $\bar{E} := \sumiinI \:
    \alpha_i \pref \Pbar_i$. We distinguish two cases.

  \medskip

  (i)~First suppose that for some $i_0\in I$ we have
  $\alpha_{i_0}=\tau$ and $\Pbar_{i_0} \brbisim \partial(\bar
  E)$. Then $\AXpb \vdash E = H + \tau.\Pbar_{i_0}$, where \plat{$H :=
    \sum_{i \in I \setminus \{i_0\}} \: \alpha_i \pref \Pbar_i$}. It
  now suffices to show that \plat{$H \sqsubseteq \Pbar_{i_0}$},
  because then axiom~\hyperlink{BP}{BP} yields $\AXpb \vdash
  \partial(E) \approx \Pbar_{i_0}$.
  So, suppose $H \arrow{\alpha}\mu$.  Then $\Ebar
    \arrow{\alpha} \mu$. Since $\partial(\Ebar) \brbisim
  \Pbar_{i_0}$, we have $\den{\Pbar_{i_0}} \Arrow{} {\nubar}
  \alphahidearrow \nu$ where $\delta(\Ebar) \brbisim {\nubar}$ and
  $\mu \brbisim \nu$.  Because $\Pbar_{i_0}$ is concrete,
  $\den{\Pbar_{i_0}} = \nubar$ by Lemma~\ref{cc}(e).  Thus
  $\den{\Pbar_{i_0}}\alphahidearrow\nu$, which was to be shown.

  \medskip

  (ii)~Next suppose that $\alpha_i\neq\tau$ or
  $P_i \nbrbisim \partial(\Ebar)$, for
  all $i \in I$, i.e., $\Ebar$ is rigid. We will show that there
  is a concrete process $\Cbar\in\calEconc$ such that
  $\AXpb \vdash \partial(\Ebar) \approx
  \partial(\Cbar)$. We proceed with induction on the number of indices
  $k\in I$ such that $\alpha_{k}\mathbin=\tau$ and $\den{P_{k}}$ can
  be written as $\mu_1\oplusr\mu_2$ with $r\mathbin\in (0,1)$ and
  $\delta(\Ebar)\brbisim\mu_1$. As $\Ebar$ is rigid, there are no
  such indices with~$r=1$.

  \smallskip

  \noindent
  \emph{Base case:} If there are no such~$k$, then $\bar C := \Ebar
  \in \calEconc$ by Lemma~\ref{cc}(a).

  \smallskip

  \noindent
  \emph{Induction step:} Let $i_0 \in I$ be an index such that
  $\alpha_{i_0} = \tau$ and $\den{P_{i_0}}$ can be written as $\mu_1
  \oplusr \mu_2$ with $r \in (0,1)$ and $\delta(\Ebar) \brbisim
  \mu_1$. First we show that it is possible to write $\den{P_{i_0}}$
  in such way while ensuring that $\mu_2$ itself cannot be written as
  $\nu_1 \probc{s} \nu_2$ with $s\in (0,1)$ and $\delta(\bar E)
  \brbisim \nu_1$. Namely, when $\mu_1 = \bigoplusjinJ \: p_j*F_j$,
  then by Lemma~\ref{cc}(g), using that $\Ebar \in \calE_r$,
  $\delta(\Ebar) = \bigoplusjinJ \: p_j*\xi_j$ with
  $\delta(F_j) \brbisim \xi_j$. So $\xi_j = \delta(\Ebar)$ for all
  $j \in J$. Now we split 
  $\den{P_{i_0}}$ in such a way between $\mu_1$ and~$\mu_2$, that for
  all $F \in \mathit{spt}(\mu_1)$ we have $F \brbisim \Ebar$ and for
  all $G \in \mathit{spt}(\mu_2)$ we have $G \nbrbisim \Ebar$. This
  ensures that $\mu_1 \brbisim \delta(\bar E)$, while $\mu_2$ cannot
  be written as $\nu_1 \probc{s} \nu_2$ with $s \in (0,1)$ and
  $\delta(\Ebar) \brbisim \nu_1$, since on the one hand
    $\mathit{spt}(\nu_1) \subseteq \mathit{spt}(\mu_2)$ and on the
    other hand $G \brbisim \Ebar$ for all $G \in
    \mathit{spt}(\nu_1)$.

   Let \plat{$H := \sum_{i\in I\setminus\{i_0\}} \: \alpha_i \pref
     \Pbar_i$}. \vspace{1pt} Then $\bar E = \tau \pref P_{i_0} + H$.
   By induction, there is a $\Cbar \in \calEconc$ with
   $\AXpb \vdash \partial(H) \approx \partial(\Cbar)$. So
   it suffices to show that $\AXpb \vdash \partial(\bar E)
   \mathbin\approx \partial(H)$. This time, this follows
   by axiom~\hyperlink{G}{G}, as soon as we obtain $\tau \pref P_{i_0}
   \mathbin\sqsubseteq \partial(H)$.
  By axioms P1--3, there are $P,P' \in \calPconc$, such that
  (a)~$\AXp\vdash P_{i_0} = P' \oplusr P$, (b)~for all $F\in
  \mathit{spt}(\den{P'})$ we have $F \brbisim \bar E$ and
  (c)~$\den{P}$ cannot be written as $\nu_1 \probc{s} \nu_2$ with $s
  \in (0,1)$ and $\delta(\bar E) \brbisim \nu_1$. Furthermore,
    by Lemma~\ref{lemma-equal-for-concrete-processes}, all $F \in
    \mathit{spt}(\den{P'})$ can be proven equal using~$\AXp$. Thus
  $\AXp\vdash P_{i_0} =\partial(F) \oplusr P$.

  We claim that $\delta(F)\arrow\tau \mu'$ for some~$\mu'$ with
    $\mu'\brbisim \den{P}$. This can be shown as follows: Since
  $\Ebar \arrow\tau \delta(F) \oplusr \den{P}$ and $\Ebar \brbisim F$,
  \vspace{1pt} we have $\delta(F) \tauhidearrow \mu$ for some $\mu
  \brbisim {\delta(F)\oplusr \den P}$. Applying the definition of
  $\tauhidearrow$, we obtain $\delta(F) \arrow\tau \mu'$ for
  some~$\mu'$ with $\mu=\partial(F) \probc{s} \mu'$ and $s \in
    [0,1]$.  Moreover, by Lemma~\ref{cc}(g), using that $\mu$ is
  concrete and thus rigid, $\mu = \eta_1\oplusr \eta_2$ with $\eta_1
  \brbisim \delta(F)$ and $\eta_2 \brbisim \den{P}$.  Since $F \in
  \calEconc$, no fraction of the distribution $\mu'$ can be branching
  bisimilar to~$\delta(F)$. Likewise, no fraction of~$\eta_2$ can be
  bisimilar to~$\delta(F)$, for then by Lemma~\ref{cc}(g) a fraction
  of~$\den{P}$ would be bisimilar to~$\delta(E)$, contradicting~(c)
  above.  Thus $s = r$ and $\mu' = \eta_2 \brbisim
  \den{P}$. This proves $\delta(F) \arrow\tau \mu'$ for some $\mu'$
  with $\mu' \brbisim \den{P}$ as claimed.

  Next, we show that $\Ebar \brbisim H \brbisim F$.
  Let $\calR$ be the smallest symmetric relation that contains $\brbisim$
  as well as $\{(\delta(H),\delta(\Ebar)), (\delta(H),\delta(F))\}$,
  and moreover satisfies $\mu_1 \mopcalR \nu_1 \land \mu_2 \mopcalR \nu_2
    \, \Longrightarrow \, (\mu_1 \oplusr \mu_2) \mopcalR (\nu_1
    \oplusr \nu_2)$. We show that $\calR$ is a branching
  bisimulation. Weak decomposability is fairly
  straightforward---compare the proof of Lemma~\ref{unrooted
    congruence}. Also trivially, it suffices to check the transfer
  condition for the pairs $(\delta(H),\delta(\Ebar))$,
  $(\delta(H),\delta(F))$, $(\delta(F),\delta(\bar H))$ and
  $(\delta(\Ebar),\delta(H))$.

  For $\delta(H) \mopcalR \delta(\Ebar)$, assume $\delta(H)
  \arrow\alpha \varrho$.  Then also $\delta(\Ebar) \arrow\alpha
  \varrho$.
  For $\delta(H)\mopcalR \delta(F)$, assume $\delta(H)
  \arrow\alpha \varrho$.  Then also $\delta(\Ebar) \arrow\alpha
  \varrho$, and since $E \brbisim F$ this move can be matched
  by~$\delta(F)$.

  For $\delta(F)\mopcalR\delta(H)$, assume $\delta(F)
  \arrow\alpha \varrho$.  Since $\Ebar\brbisim F$, we have
  $\delta(\Ebar) \alphahidearrow \eta$ for some $\eta \brbisim
  \varrho$.  Here we immediately applied Lemma~\ref{cc}(e), using that
  $\Ebar$ is rigid.  Since $F$ is concrete, no fraction of the
  distribution $\varrho$ can be branching bisimilar to $\delta(F)$, or
  to $\delta(E)$. Using that $\varrho$ is concrete, applying
  Lemma~\ref{cc}(g), this implies that no fraction of $\eta$ can be
  branching bisimilar to $\delta(E)$. Consequently, we in fact have
  $\delta(\Ebar) \alphaarrow \eta$.
  Moreover, no fraction of the transition $\delta(\Ebar)
  \alphaarrow \eta$ can stem from the transition $\Ebar \arrow\tau
  \den{P_{i_0}} = \delta(F) \oplusr \den{P}$.  It follows that
  $\delta(H) \alphaarrow \eta$ with $\eta \brbisim \varrho$.

  For $\delta(\Ebar) \mopcalR \delta(H)$, let $\delta(\Ebar)
  \arrow\alpha \varrho$.  In case $\alpha \neq \tau$, trivially
  $\delta(H) \arrow\alpha \varrho$ and we are done.  So assume~$\alpha
  = \tau$, First consider the case that $\varrho = \delta(F) \oplusr
  \den{P}$. For this case, we have shown above that
  $\delta(F) \arrow\tau \mu'$ for some~$\mu'$ with $\mu' \brbisim
  \den{P}$. Hence, by the previous case, $\delta(H) \arrow\tau \eta$
  for some~$\eta$ with $\eta \brbisim \mu' \brbisim
  \den{P}$. Consequently, $\delta(H) \tauhidearrow \delta(H)\oplusr
  \eta$. Furthermore, $\varrho = (\delta(F) \oplusr
  \den{P}) \mopcalR (\delta(H) \oplusr \eta)$.
  The more general case is that $\varrho = (\delta(F) \oplusr \den{P})
  \probc{s} \varrho'$ and $H \arrow\tau \varrho'$ for some $s \in
  [0,1]$. Now $\delta(H) \tauhidearrow (\delta(H) \oplusr \eta)
  \probc{s} \varrho'$ and $\varrho \mopcalR (\delta(H) \oplusr
  \eta) \probc{s} \varrho'$.

  Finally, we can show that  $\tau \pref P_{i_0} \mathbin\sqsubseteq
  \partial(H)$, where $\den{P_{i_0}} = \delta(F) \oplusr \den{P}$.
  This follows because $\delta(H) \tauhidearrow \delta(H) \oplusr
  \eta$ for some $\eta \brbisim \den{P}$, as argued above.
  Using that $\delta(F) \brbisim \delta(H)$, one has
  $\den{P_{i_0}}\brbisim\delta(H)\oplusr \eta$. 

  Case~(b). Suppose $P = \bigoplusiinI \: p_i * E_i$.\vspace{1pt} By case~(a) we
  can choose concrete~$\Pbar_i$ for each index $i \in I$, such that
  \plat{$\AXpb \vdash \partial(E_i) \approx \Pbar_i$}. Put $\Pbar = \bigoplusiinI \:
  p_i * \Pbar_i$. Then $\Pbar \in \calPconc$ by
  Lemma~\ref{cc}(b). By the definition of~$\approx$ it follows that
  \plat{$\AXpb \vdash {\alpha \pref P} = {\alpha \pref \Pbar}$} for
  arbitrary $\alpha \in \calA$.

  Case (c). By the proof of
  part (b) we can find concrete $\Qbar,\bar{R}\in\calPconc$ such that
  \plat{$\AXpb\vdash \alpha \pref
  Q=\alpha\pref \Qbar$} and \plat{$\AXpb\vdash \alpha\pref R = \alpha \pref
  \bar{R}$} for all $\alpha\in\calA$. Using soundness
  (Lemma~\ref{lemma-soundness}) this implies
  $\alpha\pref Q\rootedbrbisim\alpha\pref \Qbar$ and 
  $\alpha\pref R\rootedbrbisim \alpha\pref \bar{R}$,
  and hence $Q\brbisim\Qbar$ and $R\brbisim \bar{R}$.
  By Lemma \ref{lemma-equal-for-concrete-processes} we obtain from $\Qbar
  \brbisim \bar{R}$ that $\AXp\vdash \Qbar = \bar{R}$ and hence $\AXpb
  \vdash \alpha \pref \Qbar =\alpha \pref \bar{R}$ for $\alpha \in
  \calA$. This implies \plat{$\AXpb \vdash \alpha \pref Q =\alpha \pref R $}
  for $\alpha \in \calA$ and finishes the proof.
\end{proof}

\noindent
By now we have gathered all ingredients for showing that the
theory~$\AXpb$ is an equational characterization of rooted branching
probabilistic {\bisimilarity}.
It is noted that in the proof of the theorem we exploit
axiom~\hyperlink{C}{C}.

\begin{theorem}[$\!\AXpb$ sound and complete for $\rootedbrbisim$]\!
  \label{theorem-prob-completeness}
  For all non-deterministic processes $E, F \in \calE$ and all
  probabilistic processes $P,Q \in \calP$ it holds that $E
  \rootedbrbisim F$ iff $\AXpb \vdash {E = F}$ and $P \rootedbrbisim
  Q$ iff $\AXpb \vdash {P = Q}$. 
\end{theorem}

\begin{proof}
  As we have settled the soundness of~$\AXpb$ in
  Lemma~\ref{lemma-soundness}, it remains to show that \plat{$\AXpb$}
  is complete.  So, let $E, F \in \calE$ such that $E \rootedbrbisim
  F$.  Suppose $E = \sumiinI \: \alpha_i \pref P_i$ and $F = \sumjinJ
  \: \beta_j \pref Q_j$ for suitable index sets~$I, J$,
  actions~$\alpha_i, \beta_j$, and probabilistic processes~$P_i, Q_j$.

  Since, for each~$i \in I$, $E \arrow{\alpha_i}\den{P_i}$ we have
  $\delta(F) \arrow{\alpha_i} \bigoplus_{j \in J_i} \: q_{ij} * Q_j$
  and $P_i \brbisim \bigoplus_{j \in J_i} \: q_{ij} * Q_j$ for some
  subset $J_i \subseteq J$ and suitable $q_{ij} \geqslant
  0$. \vspace{2pt} Similarly, there exist for $j \in J$ a subset $I_j
  \subseteq I$ and $p_{ij} \geqslant 0$ such that $\delta(E)
  \arrow{\beta_j} \bigoplus_{i \in I_j} \: p_{ij} * P_i$ and $Q_j
  \brbisim \bigoplus_{i \in I_j} \: p_{ij} * P_i$.  By $|J| + |I|$
  series of applications of axiom~C we obtain
  \begin{align}
    & \AXp \vdash E = 
    \sumiinI \: \alpha_i \pref P_i + 
    \sumjinJ \: \beta_j \pref ( \bigoplus_{i \in I_j} \: p_{ij} * P_i )
    \text{, and} \smallskip \label{eq4} \\
    & \AXp \vdash F = \sumjinJ \: \beta_j \pref Q_j + 
    \sumiinI \alpha_i \pref ( \bigoplus_{j \in J_i} \: q_{ij} * Q_j )
    \, . 
    \label{eq5}
  \end{align}
  Since $P_i \brbisim \bigoplus_{j \in J_i} \: q_{ij} * Q_j$ and $Q_j
  \brbisim \bigoplus_{i \in I_j} \: p_{ij} * P_i$ we obtain by
  Lemma~\ref{lemma-prob-technical}
  \begin{displaymath}
    \AXpb \vdash \alpha_i \pref P_i = 
    \alpha_i \pref \textstyle\bigoplus_{j \in J_i} \: q_{ij} * Q_j 
    \  \text{and} \ 
    \AXpb \vdash \beta_j \pref Q_j = 
    \beta_j \pref\textstyle\bigoplus_{i \in I_j} \: p_{ij} * P_i 
  \end{displaymath}
  for $i \in I$, $j \in J$. Combining this with Equations (\ref{eq4})
  and~(\ref{eq5}) yields $\AXpb \vdash E = F$.

  Now, let $P, Q \in \calP$ such that $P \rootedbrbisim Q$. By
  Lemma~\ref{lemma-form-of-probabilistic-branching-processes} we have
  \begin{displaymath}
    \AXp \vdash { P = \bigoplusiinI \, p_i * E_i }
    \qquad
    \AXp \vdash { Q = \bigoplusiinI \, p_i * F_i }
    \qquad
    \forall i \in I \colon E_i \rootedbrbisim F_i
  \end{displaymath}
  for a suitable index set~$I$, $p_i > 0$, $E_i, F_i \in \calE$,
  for~$i \in I$. By the conclusion of the first paragraph of this
  proof we have $\AXpb \vdash { E_i = F_i }$ for~$i \in
  I$. Hence $\AXpb \vdash { P = Q }$.
\end{proof}



\section{Concluding remarks}

\label{sec-concl}

We presented an axiomatization of rooted branching probabilistic
{\bisimilarity} and proved its soundness and completeness. In doing
so, we aimed to stay close to a straightforward completeness proof for
the axiomatization of rooted branching {\bisimilarity} for
non-deterministic processes that employed concrete processes, which is
also presented in this paper. In particular, the route via concrete
processes guided us to find the right formulation of the axioms
\hyperlink{BP}{BP} and~\hyperlink{G}{G} for branching {\bisimilarity}
in the probabilistic case.

Future work will include the study of the extension
of the setting of the present paper with a parallel
operator~\cite{DPP05:klop}. In particular a congruence result for the
parallel operator should be obtained, which for the mixed
non-deterministic and probabilistic setting can be challenging.
\hspace{-.4pt}Also the inclusion of recursion\,\cite{DP07:tcs,FG19:jlamp}
is a clear direction for further research. 

The present conditional form of axioms \hyperlink{BP}{BP}
and~\hyperlink{G}{G} is only semantically motivated. However, the
axiom~\hyperlink{G}{G} has a purely syntactic counterpart of the form%
\footnote{The extended abstract of this paper \cite{GGV19} also proposed
 a purely syntactic counterpart of axiom B; this turned out to be incorrect,
 however.}
\begin{displaymath}
  \begin{array}{l}
    \alpha \pref \big( \,
    \partial\dg( \, \sumiinI \: \tau \pref 
    \mg(P_i \probc{r_i} \partial \dor(E+\sumiinI \tau \pref P_i\dor)
    \mg)
    + E + \sumiinI \: \tau \pref P_i\dg) \oplusr Q \, \big) 
    \smallskip \\
     \hspace{134pt}
    \qquad {} = \;
    \alpha \pref ( \, \partial \dor( E + \sumiinI \: \tau \pref P_i \dor)
    \oplusr Q \, )\;.
  \end{array}
\end{displaymath}
Admittedly, this form is a bit complicated to work with. An alternative
approach could be to axiomatize the relation~$\sqsubseteq$, or perhaps to
introduce and axiomatize an auxiliary process operator~$+'$ such that
$E \sqsubseteq P$ can be translated into the condition $E +' P = P$
or similar.

Also, we want to develop a minimization algorithm for probabilistic
processes modulo branching probabilistic {\bisimilarity}.  Eisentraut
et al.\ propose in~\cite{EHKTZ13:qest} an algorithm for deciding
equivalence with respect to weak distribution {\bisimilarity} relying
on a state-based characterization, a result presently not available in
our setting.  Other work and proposals for weak {\bisimilarity}
include~\cite{CS02:concur,FHHT16:fac,TH15:ic}, but these do not fit
well with the installed base of our toolset~\cite{Bun19:tacas}. For
the case of strong probabilistic {\bisimilarity} without combined
transitions we recently developed in~\cite{GRV18:algorithms} an
algorithm improving upon the early results of~\cite{BEM00:jcss}. In
\cite{TH15:ic} a polynomial algorithm for Segala's probabilistic
branching {\bisimilarity}, which differs from our notion of
probabilistic branching {\bisimilarity}, is defined. We hope to arrive
at an efficient algorithm by combining ideas
from~\cite{Val10:fi,VF10:tacas,TH15:ic} and
of~\cite{GV90:icalp,GJKW17:tcl}.

\bibliographystyle{plain}
\bibliography{catuscia}

\end{document}